\documentclass[reqno,11pt]{amsart}

\usepackage{amsmath, latexsym, amsfonts, amssymb, amsthm, amscd}
\usepackage[colorlinks=true, pdfstartview=FitV, linkcolor=blue, citecolor=blue, urlcolor=blue,pagebackref=false,hyperindex,breaklinks]{hyperref}
\usepackage{calc,amsfonts,amsthm,amscd,epsfig,psfrag,amsmath,amssymb,enumerate,graphicx}
\usepackage[utf8]{inputenc}
\usepackage{graphics,epsf,psfrag,epsfig}

\numberwithin{equation}{section}

\newtheorem{theorem}{Theorem}[section]
\newtheorem{lemma}[theorem]{Lemma}
\newtheorem{proposition}[theorem]{Proposition}

\newtheorem{rem}[theorem]{Remark}

\renewcommand{\ge}{\geq}
\renewcommand{\le}{\leq}
\newcommand{\ind}{\mathbf{1}}

\renewcommand{\tilde}{\widetilde}
\renewcommand{\hat}{\widehat}

\newcommand{\cT}{{\ensuremath{\mathcal T}} }

\newcommand{\bP}{{\ensuremath{\mathbf P}} }
\newcommand{\bE}{{\ensuremath{\mathbf E}} }

\DeclareMathSymbol{\leqslant}{\mathalpha}{AMSa}{"36} 
\DeclareMathSymbol{\geqslant}{\mathalpha}{AMSa}{"3E} 
\DeclareMathSymbol{\eset}{\mathalpha}{AMSb}{"3F}     
\renewcommand{\leq}{\;\leqslant\;}                   
\renewcommand{\geq}{\;\geqslant\;}                   
\newcommand{\dd}{\,\text{\rm d}}             

\newcommand{\maxtwo}[2]{\max_{\substack{#1 \\ #2}}} 

\newcommand{\bbE}{{\ensuremath{\mathbb E}} }

\newcommand{\bbN}{{\ensuremath{\mathbb N}} }

\newcommand{\bbP}{{\ensuremath{\mathbb P}} }

\newcommand{\bbR}{{\ensuremath{\mathbb R}} }

\newcommand{\bbZ}{{\ensuremath{\mathbb Z}} }

\newcommand{\gd}{\delta}
\newcommand{\gep}{\varepsilon}       

\newcommand{\gD}{\Delta}

\newcommand{\gO}{\Omega}
\newcommand{\gl}{\lambda}

\makeatletter
\def\captionfont@{\footnotesize}
\def\captionheadfont@{\scshape}

\long\def\@makecaption#1#2{%
  \vspace{2mm}
  \setbox\@tempboxa\vbox{\color@setgroup
    \advance\hsize-6pc\noindent
    \captionfont@\captionheadfont@#1\@xp\@ifnotempty\@xp
        {\@cdr#2\@nil}{.\captionfont@\upshape\enspace#2}%
    \unskip\kern-6pc\par
    \global\setbox\@ne\lastbox\color@endgroup}%
  \ifhbox\@ne 
    \setbox\@ne\hbox{\unhbox\@ne\unskip\unskip\unpenalty\unkern}%
  \fi
  \ifdim\wd\@tempboxa=\z@ 
    \setbox\@ne\hbox to\columnwidth{\hss\kern-6pc\box\@ne\hss}%
  \else 
    \setbox\@ne\vbox{\unvbox\@tempboxa\parskip\z@skip
        \noindent\unhbox\@ne\advance\hsize-6pc\par}%
\fi
  \ifnum\@tempcnta<64 
    \addvspace\abovecaptionskip
    \moveright 3pc\box\@ne
  \else 
    \moveright 3pc\box\@ne
    \nobreak
    \vskip\belowcaptionskip
  \fi
\relax
}
\makeatother
\def\writefig#1 #2 #3 {\rlap{\kern #1 truecm
\raise #2 truecm \hbox{#3}}}

\newcommand{\Tm}{T_{\rm mix}}

\title[Scaling limit of polymer pinning]
{The scaling limit of polymer pinning dynamics and a one dimensional Stefan freezing problem}

\address{CEREMADE, Place du Mar\'echal De Lattre De Tassigny
75775 PARIS CEDEX 16 - FRANCE}
\email{lacoin@ceremade.dauphine.fr}
\author{Hubert Lacoin}

\begin{document}
 
\begin{abstract}
We consider the stochastic evolution of a $1+1$-dimensional interface (or polymer) in presence of a substrate. This stochastic process
is a dynamical version of the homogeneous pinning model.
We start from a configuration far from equilibrium: a polymer with a non-trivial macroscopic height profile, and look at the evolution of a
space-time rescaled interface. 
In two cases, we prove that this rescaled interface has a scaling limit on the diffusive scale
(space rescaled by $L$ in both dimensions and time rescaled by $L^2$ where $L$ denotes the length of the interface) which we describe:
when the interaction with the substrate is such that the system is unpinned at equilibrium, then 
the scaling limit of the height profile is given by the solution of the heat equation with Dirichlet boundary condition ; when the attraction to the substrate 
is infinite, the scaling limit is given a free-boundary problem which belongs to the class of Stefan problems with contracting boundary, also referred 
to as Stefan freezing problems.
In addition, we prove the existence and regularity of the solution to this problem until a maximal time, where the boundaries collide.
\end{abstract}

\maketitle

\tableofcontents
\section{Introduction}

\subsection{The dynamical pinning model}

Random polymer models have been used for a long time by physicists to describe a large variety of physical phenomena. 
Among the numerous models that have been introduced by theoretical physicists and rigorously studied by mathematicians
(see e.g.\ \cite{cf:denH} for a survey of the most studied polymer models), 
the polymer pinning model, that involves a simple random walk interacting with a defect line, has focused a 
lot of interest
 both of the mathematics and physics community. The phase transition phenomenon between a 
pinned phase and a depinned one is now well understood, even in presence of disorder (see \cite{cf:Fisher} for a seminal paper concerning the homogeneous case, and 
\cite{cf:GB, cf:GB2} for recent reviews).

\medskip

On the other hand, dynamical pinning (which has some importance in biophysical application see \cite{cf:ABKM, cf:BKM} and references therein) 
has attracted attention of mathematicians only more recently and a lot of questions
concerning relaxation to equilibrium and its connection properties are still unsolved.

\medskip

The object of most of the mathematical studies on the dynamical pinning model up to now (see \cite{cf:CMT, cf:CLMST}) has been the mixing 
time and the relaxation time for the dynamics. It has been shown there that the mixing property of the system depends in a crucial way of 
the pinning parameter $\gl$, or more precisely, on whether the polymer at equilibrium is pinned or unpinned.

\medskip

In the present paper, we choose to study a different aspect, that is, the dynamical scaling limit of the height-profile of the polymer.
We start from an initial condition that is very far from equilibrium and approximates a macroscopic deterministic profile and we want to 
describe the evolution
of the profile under diffusive scaling. Our aim is to show that the nature of the limit of the rescaled process depends only on whether one lies in the pinned or depinned phase, and to describe 
explicitly the scaling limit in each case.

\medskip

The scaling limit is easier to guess in the unpinned phase. As in this case, there is no contact point with the substrate at equilibrium, one can infer that
the scaling limit is the same that for a system with no substrate, for which it is known that the height profile 
converges to the solution of the heat-equation
on the segment with zero Dirichlet boundary condition.

\medskip

In the pinned phase, a more interesting phenomenon takes place. In this case the dynamical picture should be the following: there are macroscopic region where 
the polymer is pinned to the substrate and other regions where the polymer stays at a macroscopic distance from it; the boundary between the pinned 
and the unpinned region is moving in time and in the unpinned region, the polymer profile evolves according to the heat-equation. The system 
reaches equilibrium when the unpinned region has totally vanished.

\medskip

In the present paper, we prove that this picture holds when the pinning parameter is infinite (or tending to infinity sufficiently fast with the size of the system).
An important step to establish this result is to prove existence and regularity of the free-boundary problem that appears in the scaling limit.


\subsection{A one dimensional Stefan freezing problem}

The free-boundary problem of unknown $(f,l,r)$ ($f$ is the function and $l$ and $r$ are the boundary of the unpinned region) that appears as the scaling limit of the pinning model in the pinned phase is the following

\begin{equation}\label{sstef}
\begin{cases} 
 \partial_t f- f_{xx}=0 \quad \text{on }  (l(t),r(t)),\\
f(\cdot,t)\equiv 0 \quad \text{on } [-1,1]\setminus (l(t),r(t)),\\
f_x(l(t),t)=- f_x(r(t),t)=1,\\
l'(t)=- f_{xx}(l(t),t),\quad  r'(t)=f_{xx}(r(t),t),\\
f(\cdot,0)=f_0,\ l(0)=l_0,\ r_0.
\end{cases}
\end{equation}
We are exclusively interested in the case of a $1$-Lipshitz initial condition that vanishes outside of the interval $(r_0,l_0)\subset [-1,1]$ and
is positive inside. This problem belongs to the class of Stefan problem, which have been introduced in mathematics to 
describe the evolution of a multiphase medium.

\medskip

The boundary condition for $f_x$ and the fact 
that the heat equation preserves Lipshitzianity imply that $f$ cannot be convex in the neighborhood of the moving boundaries,
and thus the boundary points  $l(t)$ and $r(t)$ are moving towards each other ($l'\ge 0$ and $r'\ge 0$). These problems with contracting boundary
are referred to as \textit{freezing} problems whereas those with expanding boundary are called \textit{melting} problems.

\medskip

Most of the work in the literature on Stefan problems concerns melting problems. One of the reason for this is that these problem can be rewritten
as a diffusion equation for an enthalpy function with a  diffusion coefficient that  is monotonous increasing function of the enthalpy. 
This monotonicity allows to derive uniqueness of the solution with some generality (we refer to the Introduction of \cite{cf:CS} for more precision).
The freezing problems like \eqref{sstef} on the contrary are more challenging even in the one dimensional setup. 

\medskip

Even though one dimensional freezing problems have attracted some attention in a recent past \cite{cf:CK,cf:CK2}, some amount of work is required to 
establish the existence and the unicity of a solution to \eqref{sstef}
up to a maximal time.

\section{Model and results}

\subsection{A simple model for interface motion with no constraint}\label{cornerflip}

To introduce our reader to polymer dynamics,
we first introduce the simplest version of the model where no substrate is present: this is the so-called corner-flip dynamics.
Let $\gO=\gO_{L}$ denote the set of all lattice 
paths (polymers) starting at $0$ and ending at $0$ after $2L$ steps 

\begin{equation}
\gO^{0}_L :=
\{\eta\in\bbZ^{2L+1} \ | \
\eta_{-L}=\eta_L=0\,,\;|\eta_{x+1}-\eta_x|= 1, \;x=-L,\dots,L-1\}\,.
\end{equation}
The stochastic dynamics is defined by the natural spin-flip continuous 
time Markov chain with state space $\gO^0_L$. Namely, sites $x=-L,\dots,L$ 
are equipped with independent Poisson clocks which ring with rate one: when a clock rings at $x$
the path $\eta$ is replaced by $\eta^{(x)}$, defined by $\eta_y^{(x)}=\eta_x^{(x)}$ for all $y\ne x$ and
\begin{equation}
 \eta_x^{(x)}:=
\begin{cases}
 \eta_x+2 \text{ if }  \eta_{x\pm 1}=\eta_x+1,\\
 \eta_x-2 \text{ if }  \eta_{x\pm 1}=\eta_x-1,\\
 \eta_x \text{ if } |\eta_{x+1}-\eta_{x-1}|=2.
\end{cases}
\end{equation}
One denotes by ($\tilde \eta(\cdot,t)$, $t\ge 0$), the trajectory of the Markov chain. By doing linear interpolation between the integer values of $x$, 
one can consider $\tilde \eta(\cdot,t)$ as a function of the real variable $x\in [-L,L]$.

The unique invariant measure for this dynamics is the uniform measure on $\gO^{0}_L$, and thus, from standard properties of the random walk,
when the system is at equilibrium, the rescaled interface  $(\eta(Lx)/L)_{x\in[-1,1]}$ is macroscopically flat 
($\eta(x)$ has fluctuation of order $\sqrt{L}$).
For this model, relaxation to equilibrium is well understood both in terms of mixing-time (see Wilson \cite{cf:W}) or scaling limits
(see \cite[Theorem 3.2]{cf:LST}, weaker versions of these result had been known before, using connection with the one dimensional simple 
symmetric exclusion process).

We cite in full the result concerning the scaling limit for two reasons: 
it gives some point of reference to better understand results in presence of a substrate; and we 
use it as a fundamental building brick for the proof of our new results.

\medskip

Given $f_0$ a Lipshitz function in $[-1,1]$, with $f_0(-1)=f_0(1)=0$,
set $\tilde f$ defined on $[-1,1]\times [0,\infty)$ to be the solution of the heat equation with Dirichlet boundary condition

\begin{equation}\label{ssatef}
\begin{cases}
\partial_t \tilde f- \tilde f_{xx}&=0,\\
\tilde f(\cdot, 0)&=f_0,\\
\tilde f(-1,t)&=\tilde f(1,t)=0, \quad \forall t>0.
\end{cases}
\end{equation}

\begin{theorem}[\cite{cf:LST} Theorem 3.2] \label{frefer}
\label{dynsanw}
Let $\tilde \eta^L$ be the dynamic described above, starting from a sequence of initial condition
$\eta^L_0$ that satisfies,
\begin{equation} \label{condinitcond}
\eta^L_0(x)=L(f_0(x/L)+o(1)) \text{ uniformly in $x$ when $L\to \infty$ }.
\end{equation}
Then, under diffusive scaling, $\tilde \eta^L$ converges to $\tilde f$ in law for the uniform topology: for any $T>0$, in probability,
\begin{equation}
\lim_{L\to \infty} \sup_{x\in[-1,1],t\in [0,T]} \left|\frac{1}{L}\tilde \eta^L(Lx,L^2t)-\tilde f(x,t)\right|=0.
\end{equation}

\end{theorem} 

The notation in Equation \eqref{condinitcond} means that there exists a function $\gep_L$ tending to zero such that for all $x$ and $L$
\begin{equation}
 \left|\frac{\eta^L_0(x)-L(f_0(x/L))}{L}\right|\le \gep_L.
\end{equation}
We keep  this notation for the rest of the paper.

\subsection{Dynamical polymers interacting with a substrate}

Let us now define precisely the model that is the object of study of this paper.
Our aim is to understand how the pattern of relaxation to equilibrium given by Theorem \ref{frefer} is modified (or not modified) when 
the dynamics has additional constraints. 
We focus on the case of an interface interacting with a solid substrate.
This brings us to consider a dynamics with the following modifications:
\begin{itemize}
\item We consider that a solid wall fills the entire bottom half-plane so 
that our trajectories have to stay in the positive half-plane
($\eta_x\ge 0$, $\forall x \in \{-L,\dots, L\}$). 
\item The wall interacts with the interface $\eta$ so that the rates of the transitions that modifies the number of contacts with the wall are changed.
\end{itemize}
More precisely one starts from a trajectory that lies entirely above the wall, i.e.\ which belongs to the following subset of $\gO_L^0$:
\begin{equation}
\gO_L :=
\{\eta\in\bbZ^{2L+1}\ | \;
\eta_{-L}=\eta_L=0\,; \  \forall x\in\{-L,\dots,L-1\},\ |\eta_{x+1}-\eta_x|= 1, \eta_x\ge 0 \}\,.
\end{equation}

The rates of the transitions from $\eta$ to $\eta^{(x)}$ are not uniformly equal to $1$ as in the previous section but they are given by

\begin{equation}\label{jumpr}
c(\eta,\eta^{(x)})=
\begin{cases} 
0 \text{ if } \eta_x^{(x)}=-1 \quad \text{(interdiction to go through the wall)},\\
\frac{2\gl}{1+\gl} \text{ if } \eta_x=2 \text{ and } \eta_x^{(x)}=0,\\
\frac{2}{1+\gl} \text{ if } \eta_x=0 \text{ and } \eta_x^{(x)}=2,\\
1 \text{ in every other cases }.
\end{cases}
\end{equation}

The generator $\mathcal L=\mathcal L^\gl_L$ of the Markov process is given by

\begin{equation}\label{generator}
 \mathcal L(f)=\sum_{x=-L+1}^{L-1} c(\eta,\eta^{(x)})(f(\eta^{(x)})-f(\eta).
\end{equation}

The value of the parameter $\gl\in [0,\infty]$ determines the nature of the interaction with the wall.
If $\gl>1$, the transitions adding a contact are favored, which means that the wall is attractive, whereas if $\gl<1$ 
the wall is repulsive.

\medskip

The process defined above is the {\em heat-bath} dynamics for the {\em polymer pinning model}, 
 with equilibrium measure $\pi=\pi_{2L}^\gl$ on $\gO_L$  defined by
 \begin{equation}\label{equ}
  \pi^{\gl}_{L}(\eta)
  = \frac{\gl^{N(\eta)}}{Z_{2L}^\gl}\,,
 \end{equation}
 where $$N(\eta) = \#
\{x\in[-L+1,L-1]\,:\;\eta_x=0\}$$ denotes the number of zeros in the path $\eta\in\gO$ and 
\begin{equation}
Z_{2L}^\gl := \sum_{\eta'\in \gO_L}
\gl^{N(\eta')}
\end{equation}
is the partition function, which is the renormalization factor that makes $\pi_{L}^\gl$ a probability.

For every $\gl>0$, $L\in\bbN$, $\pi_{L}^\gl$  is the unique reversible 
 invariant measure for the Markov chain. For any value of $\gl$, the 
rescaled version of $\eta$ at equilibrium (that is, under the measure $\pi_{L}^\gl$) is flat, 
but the microscopic properties of $\eta$ vary with the value of $\gl$:

\begin{itemize}
 \item [(i)] When $\gl\in [0,2)$, the interface is repelled by the wall (i.e.\ when $\gl\in (1,2)$, the entropic repulsion wins against the 
energetic attraction of the wall) 
and typical paths have a number of contacts with the wall which stays bounded when $L$ tends to infinity (the sequence of the laws of $N(\eta)$ is tight).
 \item [(ii)] When $\gl \in (2,\infty]$, the interface is pinned to the wall, and typical paths have a number of contacts with the wall which 
is of order $L$.
 \item[(iii)] When $\gl=2$, $\eta$ has a lot of contact with the wall (order $\sqrt{L}$) but the longest excursion away from the wall has length of order $L$.
\end{itemize}
For more precise statements and proofs, we refer to Chapter 2 in \cite{cf:GB}.
These three cases are respectively referred to as the depinned or unpinned phase, the pinned phase, and the critical point (or phase transition point).

\begin{rem} \rm
The case $\gl=\infty$, that will be also considered in this paper is a bit particular. 
Indeed, as seen in \eqref{jumpr}, when $\eta_x=0$ (when $\eta$ touches the wall at $x$),
it sticks to it forever, so that $\eta(\cdot,t)$ stops moving once it has reached 
the minimal configuration
$ \eta^{\min} $ defined by
\begin{equation}\label{etamin}
 \eta^{\min}_x:=\begin{cases}
                0 \text{ if } x+L \text{ is even },\\
1 \text{ if } x+L \text{ is odd },
               \end{cases}
\end{equation}
which is the configuration with the maximal number of contact point with the wall.
In that case, the unique invariance probability measure is the 
Dirac mass on $\eta^{\min}$. A question of interest is then to compute the time at 
which $\eta(\cdot,t)$ stops to move: the hitting time of $\eta^{\min}$.

\end{rem}

\medskip

Our aim is to get a result similar to Theorem \ref{dynsanw}, describing how, starting from a non-flat profile, the system relaxes to equilibrium.
We are able to deduce results in two cases:

\begin{itemize}
 \item in the depinned phase, when $\gl\in[0,2)$: in that case the scaling limit is the same one as for the model without wall.
 The result can be obtained with rather soft comparison arguments when $\gl\le 1$ but requires some additional work when $\gl\in(1,2)$.
 \item when the wall is sticky, $\gl=\infty$: in that case, the attraction of the wall can be seen ot the macroscopic level,
and the scaling limit is given by the solution of a partial differential equation with moving boundary: the free bounary problem \eqref{sstef}.
\end{itemize}

We can understand what happens when $\gl=2$ (the critical point) and when $\gl\in(2,\infty)$ (the pinned phase) at a heuristic level, and formulate 
this as conjectures (see Section \ref{secconj}). There are a lot of technical reasons why bringing these conjectures to rigorous ground cannot be done only we the ideas exposed in 
this paper. We might address this issue in future work.

\subsection{Scaling limit in the repulsive case}

Our first result is an analog of Theorem \ref{frefer} for the dynamic in the depinned phase.

\begin{theorem}\label{repcase}
Let $\eta=\eta^{L,\gl}$ be the dynamic on $\gO_L$ with generator $\mathcal L$ described above, with the parameter $\gl\in [0,2)$ 
and starting from a sequence of initial condition
$\eta^L_0$ satisfying 
\begin{equation} \label{assum}
\eta^L_0(x)=Lf_0(x/L)(1+o(1)) \text{ uniformly in $x$ when $L\to \infty$ },
\end{equation}
where $f_0$ is a $1$-Lipshitz non-negative function.

\medskip

Then $\eta^{L,\gl}$ converges to $\tilde f$ defined by \eqref{ssatef} in law for the uniform topology in the sense that for any $T>0$,
\begin{equation}
\lim_{L\to \infty} \sup_{x\in[-1,1],t\in [0,T]} \left|\frac{1}{L}\eta^L(Lx,L^2t)-\tilde f(x,t)\right|=0,
\end{equation}
in probability.
\end{theorem}

This result is not much of a surprise. For this range of $\gl$, the wall is pushing the trajectory
 $\eta$ away, so that for most of the time $\eta(t)$ lies in the wall-free zone. This is the reason why 
the effect of the wall does not appear in the scaling limit.
We believe that this also to be the case for $\gl=2$, but the fact that $\sqrt{L}$ contact with the wall
can appear at equilibrium instead of $O(1)$ brings an additional  technical challenge.

\subsection{Toward the scaling limit for pinning on a sticky substrate}

We move now to the case $\gl=\infty$. In that case (recall \eqref{jumpr}), the corners on the interface flip with rate $1$ if it does not change the number of contact with the wall, with rate $2$ if it adds one contact, and the contacts with the wall cannot be removed and stay forever.
In that case, it is known that with large probability after a time $L^2$ the dynamics ends up with a path completely stuck 
to the substrate (with probability tending to $1$), (see \cite[Proposition 5.6]{cf:CMT} or Lemma \ref{drifta} below). This implies in particular that 
the scaling limit in this case cannot 
be given by \eqref{ssatef}.

\medskip

Let us try to give some heuristic justification for the PDE problem that rules the evolution of scaling limit $f$. 
We suppose that the polymer path consists of a pinned region where it sticks to the wall and
$f \equiv 0$ and an unpinned region which corresponds to an interval $[Ll(t),Lr(t)]$ 
(i.e.\ $(l(t),r(t))$ for the rescaled picture) so that $f(t,l(t))=f(t,r(t))=0$.
In the unpinned region the wall has no influence and thus Theorem \ref{frefer} indicates that one should have $\partial_t f-f_{xx}=0$.
What is left to be determined is the speed at which the boundary of the pinned region moves (value of of the time derivative $l'(t)$ and $r'(t)$) 
and/or the boundary condition that $f$ has to satisfy at the boundary of $(l(t),r(t))$.

\medskip

\begin{rem} \rm
 With the boundary condition that one considers, $f$ is not space derivable at the extremities of $[l(t),r(t)]$. 
In what follows, when one talk about the space derivatives of $f$ at point $l(t)$, resp.\ $r(t)$, we refer to right resp.\ left derivatives.
\end{rem}

\medskip

What we want to justify here is that the slope of the scaling limit at the left (resp. right) boundary of the pinned region,
given by $f_x(l(t),t)$ (resp. $f_x(r(t),t)$) has to be equal to $+1$ (resp. $-1$).
The reason for this is that, as the scaling is diffusive, the mean-speed of the left-boundary of 
the pinned zone for the non-rescaled dynamics has to be of order $1/L$. 
This can be achieved only if the density of down-steps near the left boundary is vanishing, and hence if
 $f_x(l(t),t)=1$.

\medskip

Combining this boundary condition with $\partial_t f-f_{xx}$, and doing some trigonometry it implies (at least at the heuristic level) that one must also have $l'(t)=-f_{xx}(t,l(t))$ and $r'(t)= f_{xx}(t,l(t))$. Thus, $f$ should be the solution of 
\eqref{sstef}.

\subsection{Solving the Stefan problem}

The problem \eqref{sstef} is slightly overdetermined but this obstacle vanishes if one considers the derivative problem,

\begin{equation}\label{stef}
\begin{cases}
\partial_t \rho- \rho_{xx} \text{ in }   \quad \left(l(t);r(t)\right),\\
\rho(l(t),t)\equiv 1 \text{ on } [-1,l(t)], \quad   \rho(r(t),t)=-1 \text{ on } [l(t),1]\\
l'(t)=- \rho_x(l(t),t), \quad  r'(t)=\rho_x(r(t),t),\\
\rho(\cdot,0)=\rho_0 \text{ on } (r_0,l_0). 
\end{cases}
\end{equation}

A problem very similar to \eqref{stef} has been considered by Chayes and Kim in \cite{cf:CK} but with the third line replaced by
$$l'(t)=- \rho_x(l(t),t)/2 \text{ and } r'(t)=\rho_x(r(t),t)/2.$$ 
Note that the formulation in \cite{cf:CK} is slightly differs but this is what one finds after appropriate rescaling).
This small change has big consequences on the behavior of the solution.
Whereas for the problem considered by \cite{cf:CK}, the solution exists until a maximal time where $r(t)=l(t)$ for all 
reasonable initial condition $\rho_0$, our problem might show some degeneracy when $l$ and $r$ are still well apart (see Section \ref{discuss}).

\medskip

What we show in this paper is that this kind of complication does not occur for the initial conditions we are interested in.
Furthermore, we establish regularity and further additional properties of interest.

\begin{theorem}\label{existence}
Suppose that $f_0$ is a $1$-Lipshitz function positive and smooth on $(l_0,r_0)$ and satisfies the boundary condition 
$f'(l_0)=-f'(r_0)=1$.
Then the free boundary problem \eqref{sstef} has a unique classical solution up to time 
$$T^*=\frac{1}{2}\int^{r_0}_{l_0} f(x)\dd x,$$
at which the area below the curve vanishes.

\medskip

Furthermore:
\begin{itemize}
\item [(i)] $r(t)$ and $l(t)$ are $C^\infty$ on on the interval $(0,T^*)$,
\item [(ii)]$f$ becomes concave on $(l(t),r(t))$ before time $T^*$,
\item[(iii)] $\lim_{t\to T^*}( r(t)-l(t))=0$.
\end{itemize}

Equivalently if $\rho_0$ is the derivative of a function $f_0$ that has the above properties, then \eqref{stef}
has a solution until time $T^*$ where the two boundary meets.

\end{theorem}

The proof of short-time existence and regularity of the boundary motion strongly relies on the work and Kim and Chayes \cite{cf:CK2}, and does not require
$\rho_0$ to be the derivative of positive function, but just reasonable regularity assumption (we suppose it Lipshitz but this could be relaxed).

\medskip

The proof of existence until a maximal time uses ideas and is inspired by the work of Grayson concerning the shrinking of
curves by curvature flow \cite{cf:grayson}.

\subsection{The scaling limit for $\gl=\infty$}

We are now ready to state our result concerning the dynamics with a wall and $\gl=\infty$.

\begin{theorem}\label{mainres}
Let $\eta^{L,\infty}$ denote the dynamic with wall and $\gl=\infty$, 
and starting from a sequence of initial condition
$\eta^L_0$ which satisfies
\begin{equation} 
\eta^L_0(x)=Lf_0(x/L)(1+o(1)) \text{ uniformly in $x$ when $L\to \infty$ },
\end{equation}
where $f_0$ satisfies the assumption of Theorem \ref{existence}.
Set $f$ to be the solution of \eqref{sstef}.
\medskip
Then $\eta^{L,\infty}$ converges to $f$ in law for the uniform topology in the sense that
\begin{equation}\label{trouille}
\lim_{L\to \infty} \sup_{x\in[-1,1],t>0} \left|\frac{1}{L}\eta^{L,\infty}(Lx,L^2t)-f(x,t)\right|=0,
\end{equation}
in probability.
\medskip
Moreover, one can precisely estimate the time at which the dynamics terminates
\begin{equation} \label{ouichtrucfacile}
\mathcal T:=\inf\{ t\ge 0 \ | \ \eta(\cdot,t)=\eta^{\min}\}.
\end{equation}
We have that in probability
\begin{equation}
\lim_{L\to \infty} \frac{\mathcal T}{L^2}=\int_0^1 f_0(x)\dd x.
\end{equation}
\end{theorem}

The second part of the result \eqref{ouichtrucfacile} can be compared to the result obtained by Caputo \textit{et al.} \cite[Theorem 5.5 and Proposition 5.6]{cf:CMT}, where is proved that the
mixing time for this dynamics is of order $L^2$.

\begin{rem}\rm 
 The assumption that $f$ is strictly positive in $(l_0,r_0)$ is necessary and if $f$ cancels in the middle of the interval, the scaling limit
 depends on the microscopic details of the initial condition and the scaling limit might be a random object even for a deterministic initial condition.
\end{rem}

\subsection{Discussion on the scaling limit in the attractive case for $\gl \in (2,\infty)$}\label{secconj}

Although we are quite far from being able to prove it, we do believe that Theorem \ref{mainres} extends in some way to the whole localized phase $\gl\in (2,\infty)$.
We summarize here the full conjecture and give an idea of the technical difficulties that arises when trying to prove it.

\medskip

What we believe is that the polymer consists of one (or several) unpinned region situated a macroscopic distance of the interface, and pinned regions
where the polymers looks locally at equilibrium (recall that when $\gl>2$ the polymer has a density of contacts for $\pi^\gl_L$.
As before it is natural to say that in the unpinned region, the rescaled polymer must satisfy $\partial f-f_{xx}=0$.
However, the argument giving the slope of $f$ at the boundary of the unpinned region in the $\gl=\infty$ case is not valid when $\gl$ is finite, because
the boundary can now microscopically move in both direction as unpinning is allowed.

\medskip
To guess the value of the slope at the boundary  
$\partial_x f(t,l(t))$, $\partial_x f(t,r(t))$,
we assume that the system close to the phase separation must be in a state of local-equilibrium.

\medskip

From the equilibrium results on polymer pinning with elevated boundary condition proved in \cite{cf:LT}, 
one can infer that the equilibrium slope of a polymer at the boundary of a pinned and an unpinned phase is 
$$d_\gl:=1-\frac{2}{\gl}.$$
and that one must have 
$$ f_x(t,l(t))=-f_x (t,r(t))=d_{\gl}.$$

\medskip

Hence, the scaling limit $f$ of $\eta^{L,\gl}(x,t)$ with $\gl>2$ must satisfy the following free boundary problem

\begin{equation}\label{wcb}
\begin{cases} 
 \partial_t f- f_{xx}=0 \quad \text{on }  (l(t),r(t)),\\
f(\cdot,t)\equiv 0 \quad \text{on } [-1,1]\setminus (l(t),r(t)),\\
f_x(l(t),t)=- f_x(r(t),t)=d_{\gl},\\
l'(t)=- f_{xx}(l(t),t)/d_{\gl},\quad  r'(t)=f_{xx}(r(t),t)/d_{\gl},\\
f(\cdot,0)=f_0,\ l(0)=l_0,\ r_0.
\end{cases}
\end{equation}

The problem \eqref{wcb} presents additional technical difficulties when compared to \eqref{sstef}, even though the difference between them is just a scaling factor.
The reason for this is comes from the kind of initial condition that one wants to consider.
If one starts from an initial condition that is positive in $(l_0,r_0)$ and Lipshitz with Lipshitz constant $d_\gl$,
then Theorem \ref{existence} ensures that the solution to \eqref{wcb} exists and is well behaved until the time when the pinned region has vanished.
However, if one consider a $1$-Lipshitz boundary condition, $f_{xx}(l(t),t)$ can be positive for some times, and thus the boundary are not necessarily
contracting... which makes the problem more difficult to solve.

\medskip

When the wall is sticky, there is no loss of generality in considering that there is only one unpinned region, because two regions separated by a contact 
with the wall stay separated and behave independently.
When $\gl<\infty$ the situation is different: if one starts with two distinct unpinned regions $(l_1,r_1)$ and  $(l_2,r_2)$ with $l_2>r_1$, 
this is possible that the two region merge at a positive time.

\medskip

Besides these obstacles on the analytical side, there are many reasons why the proof of Theorem \ref{mainres}
cannot easily be adapted to the case $\gl\in(2,\infty)$. Indeed a part of the strategy relies on the control of the area below $\eta$ (see for instance Lemma \ref{drifta})
and this kind of argument seems very difficult to adapt when the polymer is allowed to detach itself from the wall.

%
%

\subsection{Stefan problem and statistical mechanics}

Free boundary problems similar to \eqref{stef} appear naturally  in thermodynamics to describe the 
motion of phase boundary in a multiphase medium (e.g.\ water in solid and liquid state), and for this reason have been the object of extensive studies
(see the seminal paper of Stefan \cite{cf:Stef} and \cite{cf:Stef2} for a survey the subject).

\medskip

There has been then some recent efforts in the area of statistical mechanics in order to prove to show that Stefan problem can be obtained as the 
limiting equation of evolution for particle systems whose microscopic behavior is random. 
Among these work we can cite
\cite{cf:CS}, where Chayes and Swindle have exhibited a particle system whose hydrodynamic limit is given by a Stefan problem,
and \cite{cf:LV} where Landim and Valle proposed a microscopic modeling of Stefan freezing/melting problem and 
proved the weak convergence of the particle density to the solution of a free boundary equation 
(for a more complete bibliography we refer to the monograph \cite{cf:KL}).

\medskip

 The main difference of the problem we study here when compared to the one considered e.g.\ in \cite{cf:CS} and \cite{cf:LV} 
 is that, microscopically,
the motion of the phase boundary does not depend only on the state of the system close to the boundary, but on the whole 
configuration (a contact to the wall can be added anywhere and not only near the boundary). This makes control of the boundary motion more difficult,
and for this reason we did not use an approach based on weak convergence like in most of the literature, but something
based on a classical interpretation of the partial differential equations.  
This is the reason why we have to prove the existence of a solution of \eqref{stef} first: we cannot prove a convergence result a priori without knowing about 
the existence of 
a classical (and sufficiently regular) solution.

\medskip

After a preliminary version of this paper was published, De Masi \textit{et al.} \cite{cf:DeMa} proved a convergence result for a special version of 
the simple exclusion process with moving sources and sinks at the boundary. They conjecture that the scaling limit they obtain is a solution of a problem similar to \eqref{sstef},
but for which there exists a stationary solution.

\subsection{Organization of the paper}

In Section \ref{stefan} we prove Theorem \ref{existence}. In Section \ref{prelim}, we introduce a few result to that will be of use both for the proof of Theorem 
\ref{repcase} and \ref{mainres}. In Section \ref{repphase} we prove Theorem \ref{repcase}.
Finally Section \ref{mainressec} contains the proof of Theorem \ref{mainres}

\section{Proof of Theorem \ref{existence}}\label{stefan}

\subsection{Decomposition of the proof}

To prove the result we proceed in three steps.
Firstly, we use results and ideas from 
\cite{cf:CK2} to prove the existence and the unicity of a $C^\infty$ solution for a short-time, 
and to extend this solution until a time where the second derivative of $f$ becomes unbounded.

\medskip

\begin{proposition}\label{shorttimes}
If $f_0$ is smooth and Lipshitz on $[l_0,r_0]$ and satisfies the boundary condition 
$f'(l_0)=-f'(r_0)=1$. Then there exists $t_1$ dor which  the problem \eqref{sstef}
has a classical solution on $[0,t_1)$ which satisfies 
\begin{equation}\label{blowup}
\lim_{t\to t_1} \sup_{x\in (l(t),r(t))} |f_{xx}(x,t)|=\infty.
\end{equation}
Furthermore $l(t)$ and r$(t)$ are $C_{\infty}$ on $(0,t_1)$.
\end{proposition}
Then we show that if $f_0$ restricted to $(l_0,r_0)$ is a concave function, 
the solution exists up to a maximal time where the boundaries $l$ and $r$ meet.

\begin{proposition}\label{neeting}
If $f_0$ is concave on $(l_0,r_0)$, then $$t_1=T^*=\int^{r_0}_{l_0} f_0(x)\dd x/2$$ 
and
$$\lim_{t\to T^*}(r(t)-l(t))=0.$$
\end{proposition}
Finally, we show that if $f_0$ is positive, then it becomes concave 
before $t_1$.

\begin{proposition}\label{concavification}
If $f_0$ is positive, then there exists a time $t_2<t_1$ such that 
$f(\cdot,t)$ is concave for $t\in[t_2,t_1)$.
\end{proposition}

Theorem \ref{existence} is obtained by combining the three statements together.

\subsection{Discussion on the case of general $f_0$}\label{discuss}
Before going into the proofs, let us first discuss about why the positivity of $f$ is required for the existence of a solution until $T^*$.
We show in this section with simple examples that if this condition is violated, the solution might degenerate and the boundary condition
can stop being satisfied before $l(t)$ and $r(t)$ meet.

\medskip

First note that the fact that the area below the curve vanishes at $T^*$ is a simple consequence of 

\begin{equation}
\partial_t \left(\int_{l(t)}^{r(t)} f(x,t) \dd x\right)=\int_{l(t)}^{r(t)} f_{xx}(x,t) \dd x
=f_x(l(t),t)-f_x(r(t),t)=-2.
\end{equation}
so that for all $t\in(0,t_1)$ for which the solution is defined
\begin{equation}\label{area}
\int_{l(t)}^{r(t)} f(x,t) \dd x=\int_{l(t)}^{r(t)} f_0(x) \dd x-2t.
\end{equation}

Thus if $f_0$ is such that $\int_{l(t)}^{r(t)} f_0(x)<0$ and satisfies the assumption of Proposition \ref{shorttimes}
the signed area $\int_{l(t)}^{r(t)} f(x,t) \dd x$ is bounded away from zero uniformly in time, and it implies that the solution must degenerate before 
$l(t)$ collides with $r(t)$. If it were not the case, the fact that $f$ is Lipshitz would imply 
$$\lim_{t\to t_1} \int_{l(t)}^{r(t)} f(x,t) \dd x=0$$ which is impossible.

\medskip

In fact, even if $\int_{l(t)}^{r(t)} f(x,t) \dd x$ is initially positive, there is no guarantee that the solution exists until $l$ meets $r$.
We give a simple counter-example here:
set $l_0=-3\pi/2$, $r_0=3\pi/2$ (we choose $[l_0,r_0]$ to be the interval of definition instead of $[-1,1]$), and set
$$f_0(x):=-\cos(x).$$
The function $f_0$ satisfies the assumption of Proposition \ref{shorttimes} which implies that
a solution to \eqref{sstef} exists until a positive time $t_1$ for which the second derivative explodes.

\medskip

A quick analysis of the problem shows that as long as $f$ is negative somewhere, the graph of $f$ consists of two positive bumps that frame a negative one (see Figure \ref{dbumps}). We call the unique interval where $f$ is negative$(z_1(t),z_2(t))$.

\medskip

The geometric area of the negative bump satisfies

\begin{equation}
 \partial_t\left(\int_{z_1(t)}^{z_2(t)}|f(x,t)|\dd x\right)=f_x(z_1(t),t)-f_x(z_2(t),t)> -2,
\end{equation}
where the last inequality comes from the fact that $|f_x|<1$ in $(r(t),l(t))$ for all positive time.
Hence, if the solution continues to exists, $f$ stays negative somewhere at least until a time $t'>\left(\int_{-\pi/2}^{\pi/2}|f_0(x)|\dd x\right)/2=1$.
However, from Equation \eqref{area}, for $t=t'$ 
\begin{equation}
 \int_{l(t^*)}^{r(t^*)} f(x,t^*)\dd x =2-2t^*<0,
\end{equation}
meaning that $l$ and $r$ cannot meet for $t\ge t^*$.

\begin{figure}[hlt]
 \begin{center}
 \leavevmode 
 \epsfxsize =8.5 cm
 \psfragscanon
 \psfrag{t=0}{\tiny $t=0$}
 \psfrag{t=t1}{\tiny $t=t_1$}
\psfrag{z1}{}
\psfrag{z2}{}
\psfrag{l0}{\tiny $-3\pi/2$}
\psfrag{r0}{\tiny $3\pi/2$}
 \epsfbox{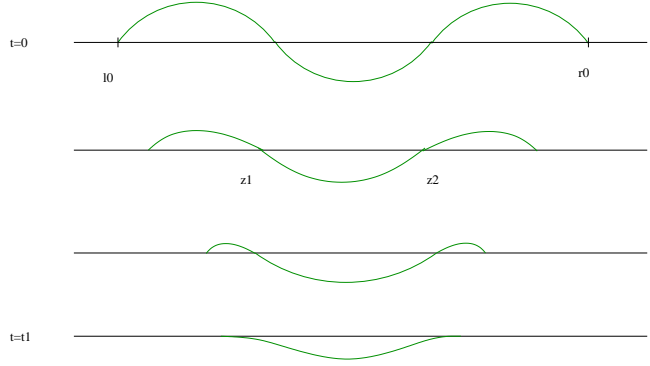}
 \end{center}
 \caption{\label{dbumps} A simple example of an initial condition for which the solution stops to exits before the boundary meets.}
 \end{figure}

\medskip

In fact a closer inspection to the proof of Proposition \ref{neeting} and \ref{concavification} reveals that 
at time $t_1$, when the second derivative $\| f_{xx}\|_{\infty}$ explodes, the two external positive bumps disappear and $f$ becomes completely negative.
As the boundary condition for $f_x$ is not satisfied anymore, it is impossible to define any reasonable notion of solution for $t\ge t_1$.

\medskip

We refer the reader to \cite[Section 4.3.]{cf:CK2} for some additional discussion on what can occur after $t_1$ for the derivative problem \eqref{stef} with different boundary conditions.

\subsection{Proof of the short-time existence: Proposition \ref{shorttimes}}
Instead of solving \eqref{sstef} directly, we solve the corresponding derivative problem \eqref{stef}.

Given a solution $(\rho,l,r)$ to \eqref{stef} with initial condition $\rho_0=f'_0$,
the triplet $(f,l,r)$ with $f$ 
defined by
\begin{equation}
 f(x,t)=\int_{l(t)}^{r(t)}\rho(x,t)\dd x,
\end{equation}
is a solution of \eqref{sstef}. 
Indeed, the initial condition is satisfied and we have
\begin{equation}\begin{split}
\partial_t f(r(t),t)&=-r'(t)+l'(t)+\int_{l(t)}^{r(t)}\rho_{xx}(x,t)\dd x=0,\\
\partial_t f(x,t)&=-r'(t)+\int_{l(t)}^{x}\rho_{xx}(x,t)\dd x=\rho_x(x,t)=f_{xx}(x,t) \ \text{for} \ x\in(l(t),r(t)).
\end{split}\end{equation}

To solve \eqref{stef}, we adapt a method developed in \cite{cf:CK2} for a similar contracting one dimensional Stefan problem.
In what follows we denote by $\rho_0$ the initial condition $f_0$.

\medskip

Let us consider two (fixed) continuous functions $\tilde l(t)$ and $\tilde r(t)$, $\tilde l$ increasing, $\tilde r$ decreasing, 
that satisfy $\tilde l(0)=l_0$ and $\tilde r(0)=r_0$. 
We consider the heat equation on the contracting domain 
$$(\tilde l,\tilde r):=\cup_{t\le T} (\tilde l(t),\tilde r(t))\times \{ t \}$$
with fixed boundary condition $+1$ on the left and $-1$ on the right,

 \begin{equation}\label{contract}
\begin{cases}
\partial_t \rho -\rho_{xx}=0  \text{ on }   \quad \left(\tilde l(t); \tilde r(t)\right),\\
\rho(\tilde l(t),t)=-\rho(\tilde r(t),t)=1,\
\rho(0,\cdot)=\rho_0.
\end{cases}
\end{equation}

We need the following technical result.

\begin{lemma}\label{lebarel} 

Given $\rho_0$ derivable that satisfies the boundary condition, and setting $\bar l=(4\|\rho'_0\|)^{-1}$.
If $\tau$ is such that 
\begin{equation}
\forall t<\tau,\  \tilde l(t)\le l_0+2\bar l, \quad \text{and} \quad  \tilde r(t)\ge r_0-2\bar l,
\end{equation}
 then the solution $\rho$ of \eqref{contract} satisfies
 \begin{equation}\begin{split}
 \rho(t,x)&\ge 0, \quad \forall t\le \tau,\ \forall x\in[\tilde l(t),l_0+2\bar l], \\
  \rho(t,x)&\le 0, \quad forall t\le \tau,\ \forall x\in[r_0-2\bar l,\tilde r(t)] .  
  \end{split}
 \end{equation}
 \end{lemma}

Of course $\bar l$ depends on the length the interval $[l_0,r_0]$, but, as $\rho_0$ satisfies the prescribed
boundary condition, one has $$r_0-l_0\ge 2(\|\rho'_0\|_{\infty})^{-1}=8\bar l.$$ 
\begin{proof}
By symmetry we can restrict ourselves to the first statement.
Because of the boundary condition, the solution that we have to consider is larger than the solution of the heat-equation on $[l_0,l_0+4\bar l]$
with initial condition $\rho_0$ and Dirichlet boundary condition $+1$ at $l_0$ and $-1$ on $l_0+4\bar l$. 
Then we notice that 
$$\rho_0(x)\ge 1-2(x-l_0)\|\rho'_0\|_{\infty}=\rho_{\min}(x), \quad \forall x\in [l_0,l_0+4\bar l].$$
Finally, we remark that $\rho_{\min}(x)\ge 0$ on $[l_0,l_0+2\bar l]$ and that as $\rho_{\min}$ is a stationary solution of the heat-equation on $[l_0,l_0+4\bar l]$ mentioned above, the inequality remains valid for all further time.

\end{proof}

Given $\tau$ that satisfies 
the assumption of Lemma \ref{lebarel} and  let us consider
\begin{equation}\label{lesphi}
\mathcal J:=\left\{  (\phi_1,\phi_2) \in \left(L^\infty([0,t_0] \right)^2 \ | \ \forall t\in[0,t_0],\ 
\phi_1(t)\in [0,1], \ \phi_2(t)\in[-1,0]\right\}.
\end{equation}

We are going to construct an application
$$\Phi: \mathcal J\to \mathcal J,$$
in two steps.

\medskip

First, given $(\phi_1,\phi_2)\in \mathcal J$, we define two one-sided contracting Stefan problems, with respective initial domains $[l_0,l_0+2\bar l]$ and 
$[r_0-2\bar l,r_0]$, and an imposed boundary condition on the non moving side: $l_0+2\bar l$ and $r_0-2\bar l$ respectively, given by $\phi_1$ and $\phi_2$.

\begin{equation}\label{grimoire}\begin{split}
\begin{cases}
\partial_t \rho^{(1)}- \rho^{(1)}_{xx}=0 \quad  \text{ on }   \left(l(t); l_0+2\bar l\right),\\
\rho^{(1)}(l(t),t)=1,  \quad  l'(t)=-\rho^{(1)}_x(l(t),t),\\
\rho^{(1)}(l_0+2\bar l,t)=\phi_1(t), \\
\rho^{(1)}(x,0)=\rho_0(x) \text{ on } [l_0,l_0+2\bar l]
\end{cases}
\\
\begin{cases}
\partial_t \rho^{(2)}-\rho^{(2)}_{xx}=0 \quad \text{ on }    \left(r_0-2\bar l,r(t)\right),\\
\rho^{(2)}(r(t),t)=-1,  \quad  r'(t)=\rho^{(2)}_x(r(t),t),\\
\rho^{(2)}(r_0-2\bar l, t)=\phi_2(t), \\
\rho^{(2)}(x,0)=\rho_0(x) \text{ on } [r_0-2\bar l,r_0].
\end{cases}
\end{split}
\end{equation}

According to \cite[Theorem 2.2]{cf:CK}, these two problem have a solution until the time when $l(t)$ meets $l_0+2\bar l$ or $r(t)$ meets $r_0-2\bar l$ respectively
(more precisely, to fit exactly the setup of \cite{cf:CK} where $\rho\equiv 0$ on the moving boundary, one must consider the problems
solved by $1-\rho^{(1)}$ and $1+\rho^{(2)}$), and the boundary $r$ and $l$ are $C^\infty$ on $(0,T)$ where $T$ is the time where the solution ceases to exist. 

\medskip

For technical purpose we want to guarantee that for any choice of $(\phi_1,\phi_2)$, the boundaries $l$ and $r$ need some time to come 
half-way towards the fixed boundary.

\begin{lemma}\label{away}
Let $(\rho^{(1)},l)$ and $(\rho^{(2)},r)$ be solutions of \eqref{grimoire} with boundary condition $(\phi_1,\phi_2)\in \mathcal J$, and initial condition $\rho_0$.
There exists a universal constant $c$ such that for all $t\le c \|\rho'_0\|^2$,

\begin{equation}
l(t)\le l_0+\bar l \quad  \text{and} \quad r(t)\ge r_0-\bar l.
\end{equation}

\end{lemma}

\begin{proof}
By symmetry it is sufficient to perform the proof only for $l(t)$. We suppose also that $l_0=0$.
Set for $x\in(l(t),2\bar l]$,
$$f^{(1)}(x,t)=\int_{l(t)}^{x}\rho^{(1)}(x,t)\dd x.$$

The reader can check that $(f^{(1)},l)$ is a solution of 
\begin{equation}\label{integration}
\begin{cases}
\partial_t f^{(1)}- f^{(1)}_{xx}=0 \quad  \text{ on }   \left(l(t); l_0+2\bar l\right),\\
f^{(1)}(l(t),t)=0,  \quad  l'(t)=-f^{(1)}_{xx}(l(t),t),\\
f_x^{(1)}(l_0+2\bar l,t)=\phi_1(t), \quad f_x^{(1)}(l(t),t)=1,\\
f^{(1)}(x,0)=f_0(x)  \text{ on } [l_0,l_0+2\bar l].
\end{cases}
\end{equation}
The solution $f^{(1)}$ is monotone in $f_0$ and $\phi_1$: it decreases if $f_0$ and/or $\phi_1$ are decreased.
Thus $l(t)$ is smaller than $l_{\min}(t)$ which is obtained by taking $\phi_1\equiv 0$ and $f^{(1)}(\cdot,0)$ to be the integral of $\rho_{\min}$ 
\begin{equation}
f^{(1)}(\cdot,0)=x-\| \rho'_0\|x^2.
\end{equation}
By diffusive scaling 
$$\inf\{ t\  | \  l_{\min}=\bar l \}=c\| \rho'_0\|^{-2}_{\infty}.$$
\end{proof}

Given $\phi_1$ and $\phi_2$, we generate $r(t)$ and $l(t)$ that are given by the solutions of problems \eqref{grimoire} with initial condition 
$\rho_0$.
Then we define $\bar \rho$ to be the solution of the heat equation on the contracting domain 
$$(l,r):=\cup_{t\le T} (l(t),r(t))\times \{ t \}$$ with $+1$ $-1$ boundary condition and set 
$$\Phi(\phi_1,\phi_2)(t):=\left(\bar \rho(l_0+2\bar l,t), \bar \rho(r_0-2\bar l,t)\right).$$

The fundamental building brick of our proof is the following adaptation of \cite[Proposition 3.1]{cf:CK2}.
\begin{proposition}\label{letzero}
There exists $t_0(\|\rho_x\|_{\infty})$ such that the function $\Phi: \mathcal J \to \mathcal J$ is a contraction map 
for the $l_\infty$ norm. 
In other words, there exists $m(t_0)<1$, such that 
for every $(\phi_1,\phi_2), (\bar \phi_1,\bar \phi_2)\in \mathcal J$,

\begin{equation}
\sup_{t \in [0,t_0]} | \Phi(\phi_1,\phi_2)(t)-\Phi(\bar \phi_1,\bar \phi_2)(t)|\le m(t_0)|(\phi_1,\phi_2)(t)-(\bar \phi_1,\bar \phi_2)(t)|
\end{equation}
where $|\cdot|$ stands for the $l_\infty$ norm in $\bbR^2$.
\end{proposition}
\begin{proof}
See  \cite[Proof of Proposition 3.1 and of Theorem 3.4.]{cf:CK2}.
Lemma \ref{away} is needed as one need that for small time $l(t)$, and $r(t)$ stay at a positive distance from the fixed boundary
uniformly in  $(\bar \phi_1,\bar \phi_2)\in \mathcal J$
Note that in Proposition 3.1, the time $t_0$ is said to depend of four different quantities, but the reader can check that they can all be expressed in term of $\|\rho'_0\|_{\infty}$.
\end{proof}

The solution of \eqref{stef} until time $t_0$ 
choose $l(t)$ and $r(t)$ to be moving boundary condition generated by the problems \eqref{grimoire} with initial condition $\rho_0$ 
and boundary condition $(\phi_1,\phi_2)$ given by the unique fixed point of $\Phi$.  Then from the definition of $\Phi$, the solution of the heat equation in the contracting domain
$(l,r)$ satisfies the boundary condition of \eqref{stef}.
The solution is unique because of unicity of the fixed point of the contraction and is smooth for positive times because as stated in
\cite[Theorem 2.2]{cf:CK} the one sided problems \eqref{grimoire} generate boundaries that are $C_\infty$ for positive time.

\medskip

Let us now show that the solution can be extended until the derivative explodes.
Suppose that the solution exists and is smooth until a time $t_2$ and that 

$$\sup_{t< t_2} \max_{x\in (l(t),r(t))} |\rho_{x}(x,t)|\le K<\infty$$

Then taking the $t_0$ corresponding to $K$ in Lemma \ref{letzero}, we can take the solution at time $t_2-(t_0/2)$
and extend it until time $t_2+(t_0/2)$. Smoothness of the motion of the boundary at the time $t_2-t_0/2$ is guaranteed by unicity of short time solutions.
Iterating this procedure, we see that the maximal time for which a smooth solution exists must satisfies \eqref{blowup}.
\qed
\medskip

\subsection{The convex case: proof of Proposition \ref{neeting} when $f$ is concave}

We introduce the notation 
$$k(x,t):=-f_{xx}(t).$$
When $f$ is concave $k$ is positive but as in this section we prove some technical results that are also valid also when $k$ is allowed to be negative, we will 
mention to the reader when we suppose that $k$ is positive.

\medskip

The line of the proof is to show first  that around time $t_1$, we must have 
$$\lim_{t \to T_1} \int_{l(t)}^{r(t)} k\log k=\infty,$$
and then to show that $ \int_{l(t)}^{r(t)} k\log k \dd x$ can be large only if the area below the graph of $f$ is small.
Recalling Equation \eqref{area}, the area is small only if $t$ is close to $T^*$. The combination of these two statements 
implies that one must have $t_1=T^*$.
The inspiration for many ingredients of this proof comes from \cite{cf:grayson}.

\medskip

The first point can be stated as follows

\begin{lemma}\label{ilestborn}
At any positive time $t<t_1$, we have
\begin{equation}
 \|k(\cdot,t)\|_{\infty}\le \max\left(  \|k(\cdot,0)\|_{\infty}, k(t,r(t)), k(t,l(t))\right).
\end{equation}
and also
\begin{equation}
\max_{s\le t} \|k_x(\cdot,t)\|_{\infty}\le  \max_{s\le t}\left(  \|k_x(\cdot,0)\|_{\infty}, k(s,r(s))^2, k(s,l(s))^2\right).
\end{equation}
As a consequence
\begin{equation}
\limsup_{t\to t_1} \int^{r(t)}_{l(t)} (k\log k)\ind_{k\ge 0} \dd x=\infty.
\end{equation}
\end{lemma}

To prove the result, we notice that 
the boundary constraint for our problem yields a simple relation between $k$ and its derivative at the boundary.

\begin{lemma}\label{orreur}
We have for all $t\in(0,t_1)$.
\begin{equation}\label{relationone}
k_x(l(t),t)=-k^2(l(t),t) \text{ and } k_x(r(t),t)=k^2(r(t),t).
\end{equation}
\end{lemma}
\begin{proof}
The result is obtained by derivating the equality $f_x(t,l(t))=1$ which gives
\begin{equation}
\partial t \left(f_x(l(t),t)\right)= -l'(t) k(l(t),t)- k_x(l(t),t)=-(k_x+k^2)(l(t),t)=0.
\end{equation}
The proof for $r(t)$ is similar.
\end{proof}

\begin{proof}[Proof of Lemma \ref{ilestborn}]
Inside the interval $(l,r)$, $k$ evolves according to the heat equation. This implies that its local maxima decrease and its local minima increase.
Furthermore, equation \eqref{relationone} guarantees that $k_x$ never vanishes on the boundary at positive 
times so that local extrema cannot be created from the boundary. Hence at all times, local minima and maxima of $k$ in $(l,r)$ 
are smaller than $ \|k(\cdot,0)\|_{\infty}$ in absolute value. Thus, if the overall maximum is larger than  
$\|k(\cdot,0)\|_{\infty}$,
it must be reached on the boundary.

\medskip

For the derivative $k_x$ there is no easy argument that prevents 
the creation of new local maxima $k_x$ from the boundary (although we do not believe it can occur). 
However if $t$ is such that 

\begin{equation}
\|k_x(\cdot,t)\|_{\infty}=\max_{s\le t} \|k_x(\cdot,s)\|_{\infty},
\end{equation}
the fact that local maxima of $k_x$ in $(l,r)$ decrease and local minima increase implies that 
the overall maximum of $k_x$ is realized on the boundary. This yields the second result.

\medskip

Finally, consider $K\ge \max( \|k(\cdot,0)\|_{\infty}, \sqrt{\|k_x(\cdot,0)\|_{\infty}})$,
and $t_K$ the first time where $\|k(\cdot,t)\|_{\infty}=K$. From the two first points,
the maximum of $k$ is reached on the boundary of $[l,r]$, and thus $k_x$ is bounded by $K^2$.
It yields the following inequality

\begin{multline}
\int^{r(t_k)}_{l(t_k)} (k\log k)\ind_{k\ge 0} \dd x\ge e^{-1}(r_0-l_0)
\\
+\int_{l(t_K)}^{l(t_K)+\frac{1}{K}}(K-K^2(x-l(t_K)))\log (K-K^2(x-l(t_K)))\dd x.
\end{multline}
 After a change of variable the reader can check that second term is equal to $\frac{1}{2}\left(\log K-1/2\right)$, 
 which allows us to conclude by letting $K$ go to infinity.

\end{proof}

Now are ready to prove that $t_1=T^*$ when $f_0$ is concave. To do so, we suppose that $t_1\le T^*-\gep$ for some positive $\gep$ and show that this implies 
\begin{equation}\label{boundedness}
\int^{r(t)}_{l(t)} (k\log k) \dd x  \text{ is uniformly bounded on } [0,t_1),
\end{equation}
which contradicts Lemma \ref{ilestborn}.

\medskip

Our strategy for the proof is to show that until time $t_1$ the time derivative of $\int^{r(t)}_{l(t)} (k\log k)\dd x$ is uniformly bounded.

\begin{multline}\label{bypartitis}
\partial_t\left(\int^{r(t)}_{l(t)} k \log k \dd x \right)\\
= -l'(t)(k\log k)(l(t),t)+r'(t)(k\log k)(r(t),t)+ \int^{r(t)}_{l(t)}\partial_t(k \log k)\dd x
\\=-(k^2\log k)(l(t),t)-(k^2\log k)(r(t),t)+
\int^{r(t)}_{l(t)} k_{xx} (\log k+1)\dd x \\
= -(k^2\log k)(l(t),t)-(k^2\log k)(r(t))+[k_x(\log k+1)]_{l(t)}^{r(t)}-\int^{r(t)}_{l(t)} \frac{k^2_{x}}{k}\dd x ,
\end{multline}
where the third inequality is obtained by using integration by parts.
The relation \ref{relationone} between $k$ and its derivative makes some of the terms cancel each other and 
we end up with

\begin{equation}\label{marmouse}
\partial_t\left(\int^{r(t)}_{l(t)} k \log k \dd x\right)=k^2(l(t),t)+k^2(r(t),t)-\int^{r(t)}_{l(t)} \frac{k^2_{x}}{k}\dd x.
\end{equation}

In order to show that the r.h.s.\ is uniformly finite, 
we first show that it is not possible to have, on the graph of $f$, $k$ large  on an arc whose total curvature is close to
$\pi/4$ near one of the boundaries. If it were not the case, the concavity of $f$ would imply that the area below the graph of $f$ is small.
We need to introduce some definitions. Set 

\begin{equation}\label{ak}\begin{split}
a(K,t)&:=\inf\{x\ge l(t) \ | \  k(x,t) \le K^2\}, \\
b(K,t)&:=\sup\{x\le r(t) \ | \  k(x,t) \le K^2\}.
\end{split}
\end{equation}

\begin{lemma}\label{akk}
We have for all $t$,
$$f(a(K,t),t)\le K^{-2}.$$
In addition, if
$\int_{l(t)}^{a(K,t)} k \dd x\ge 1-\gd$
and $f$ is concave,
\begin{equation}
\label{gijoe}
\int_{l(t)}^{r(t)} f \dd x\le \left(K^{-2}+\gd(r_0-l_0)\right)(r_0-l_0).
\end{equation}
 If $\int_{b(K,t)}^{r(t)} k \dd x \le 1-\gd$ then $f(B(K,t),t)\le K^{-2}$ and the inequality \eqref{gijoe} also holds.
\end{lemma}

As a consequence, we get that if $K$ and $\delta$ are sufficiently large resp.\ small 
so that $$ \left(K^{-2}+\gd(r_0-l_0)\right)(r_0-l_0)\le 2\gep,$$
then for all $t< t_1\le T^*-\gep$,
\begin{equation}\label{grosarc}
\int_{l(t)}^{a(K,t)} k \dd x< 1-\gd \quad \text{ and } \quad \int_{b(K,t)}^{r(t)} k \dd x < 1-\gd,
\end{equation}
because if not, the conclusion of Lemma \ref{akk}  would contradict \eqref{area}.

We are going to use this information to bound the r.h.s.\ of \eqref{marmouse} from above, 
by using the following functional inequality sometimes referred to as Agmon's inequality.
We include its proof at the end of the section for the sake of completeness.

\begin{lemma}\label{agmon}
Let $\gamma$ be a function in $L_2(\bbR_+)$ whose derivative is in  $L_2(\bbR_+)$.
Then
\begin{equation}\label{laborn}
\| \gamma \|^4_{\infty}\le 4 \left(\int_{\bbR_+} \gamma^2 \dd x\right)  \left(\int_{\bbR_+}\gamma_x^2 \dd x\right).
\end{equation} 

\end{lemma}
We apply the inequality to $\gamma$ defined as
\begin{equation}\label{reptile}
\gamma(x)=
\begin{cases}
\sqrt{k(x+l(t),t)}-K &\text{ if }  x\le a(K,t)-l(t),\\
0 &\text{ if }  x\ge a(K,t)-l(t).
\end{cases}
\end{equation}
With this definition, \eqref{grosarc} reads

$$\int_{\bbR_+} \gamma^2 \dd x\le \int_{l(t)}^{a(K,t)} k\dd x \le 1-\delta.$$
Then we obtain that 
\begin{equation}
\int^{a(K,t)}_{l(t)} \frac{k^2_{x}}{k}\dd x=4\int_{l(t)}^{a(K,t)} \gamma^2_{x}\dd x\ge \frac{\| \gamma \|^4_{\infty}}{\int_{\bbR_+} \gamma^2 \dd x}
\ge \frac{\left(k(l(t))^{1/2}-K\right)^4}{1-\gd}.
\end{equation}
If $k(l(t))\ge 16 K^2/\gd^2$ this is larger than $k^2(l(t))$, and in any case it is non-negative.
Hence 
\begin{equation}\label{croconodile}
 k^2(l(t))-\int^{a(K,t)}_{l(t)} \frac{k^2_{x}}{k}\dd x\le 256 K^4/\gd^4.
\end{equation}
Symmetrically
\begin{equation}
 k^2(r(t))-\int^{r(t)}_{b(K,t)} \frac{k^2_{x}}{k}\dd x\le 256 K^4/\gd^4,
\end{equation}
and hence, combining these inequalities with \eqref{marmouse}, we have

\begin{equation}
\partial_t\left(\int^{r(t)}_{l(t)} k \log k\right)\le 512 K^4/\gd^4.
\end{equation}
This implies that $\int^{r(t)}_{l(t)} k \log k$ remains bounded, and gives a contradiction to Lemma \ref{ilestborn}.

\medskip

To finish the proof of Proposition \ref{neeting}, we show now that, when $f_0$ is concave, $$\lim_{t\to T^*} r(t)-l(t)=0.$$
Because of \eqref{area} and the fact that $f$ is a Lipshitz function, 
we have
\begin{equation}\label{lemax}
 \max_{x\in[l,r]} f(x,t)\le \sqrt{2(T^*-t)}.
\end{equation}

Since $\min_{x\in[l,r]} k(\cdot,t)$ is an increasing function of time, for $t\ge T^*/2$, 
$k$ is uniformly bounded away from zero, say $k\ge \eta>0$. This combined with with \eqref{lemax} implies that
\begin{equation}
 (r-l)(t)\le 2 \left(\frac{8(T^*-t)}{\eta}\right)^{1/4}.
\end{equation}

\qed

We end this section with the proof of Lemma \ref{akk} and Lemma \ref{agmon}

\begin{proof}[Proof of Lemma \ref{akk}]
As $f$ is a Lipshitz function $\int^{a(K,t)}_{l(t)} k \dd x =1-f_x(a(K,t),t)\le 2$.
Hence, using again Lipshitzianity and the definition of $a(K,t)$,
\begin{equation}
 f(a(K,t),t)\le a(K,t)-l(t)\le K^{-2}\int^{a(K,t)}_{l(t)} k \dd x  \le 2K^{-2}.
\end{equation}
If $f$ is concave and $\int^{a(K,t)}_{l(t)} k \dd x\ge 1-\delta$, 
then  
$$\forall x\ge a(K,t),\ f_x(x,t) \le f_x(a(K,t),t)=\delta,$$
 and hence for all $x\ge (a(K,t))$ 
 \begin{equation}
  f(x,t)\le f_x(a(K,t),t)+\delta(x-a(K,t))\le 2K^{-2}+\delta(r_0-l_0).
 \end{equation}
As the bound also holds for $x\le a(K,t)$, we can integrate the inequality over $[r(t),l(t)]$ to conclude.

\end{proof}

\begin{proof}[Proof of Lemma \ref{agmon}]
 It is sufficient to prove the result when the maximum of $\gamma$ is attained at $0$ (if it is attained at a positive value $x_0$, we can then consider 
 $\gamma( \cdot -x_0)$ restricted to $\bbR_+$ which has the effect of making the r.h.s.\ of \eqref{laborn} smaller).
 As the inequality is invariant by the scalings $\gamma\to \gl \gamma$ and $\gamma\to \gamma( \gl \ \cdot)$ with $\gl\in(0,\infty)$,
 we can also assume that $\int \gamma^2=\int \gamma_x^2=1$.
Then we have
 
 \begin{equation}
  0\le \int (\gamma-\gamma_x)^2\dd x=\int \gamma^2\dd x+ \int \gamma_x^2\dd x -\gamma(0)^2,
 \end{equation}
and hence $\gamma(0)^4\le 4$.

\end{proof}

\subsection{The concavification of positive initial condition: proof of Proposition \ref{concavification}}

Let us now move to the non-convex case.
We consider the inflection points on the graph of $f$, that is, the points around which $k$ changes sign. When $t>0$ there are only finitely many of them,
and furthermore, their number is decreasing in time and they move continuously.
We want to show that the last inflection point disappears before $t_1$. 

\medskip

We suppose that this is not the case, and then we show that $\int^{r(t)}_{l(t)} k \log k\ind_{k\ge 0}$ remains bounded
when $t$ approaches $t_1$. The first thing we do is to place ourselves in a neighborhood of $t_1$ where the number of inflection points is constant 
(it has to be even because they are the extremities of arcs where $k$ is negative).
Let $i_1,\dots,i_{2p}$ denote the abscissa of these inflection points. Then 
\begin{equation}
\int^{r(t)}_{l(t)} k \log k\ind_{k\ge 0}\dd x= \int^{r(t)}_{i_1(t)} k \log k\dd x +\sum_{j=1}^{p-1}   \int_{i_{2j}(t)}^{i_{2j+1}(t)} k \log k \dd x+\int_{i_{2p}}^{r(t)} k \log k \dd x.
\end{equation}
First notice that using integration by part, we have

\begin{equation}\label{bolide}
 \partial_t\left( \int^{i_{2j}(t)}_{i_{2j+1}(t)} k \log k \dd x\right)= -\int^{i_{2j}(t)}_{i_{2j+1}(t)} \frac{k^2_x}{k} \dd x \le 0.
\end{equation}
Hence to prove that $\int^{r(t)}_{l(t)} k \log k\ind_{k\ge 0} \dd x$ remains bounded we just have to check that the extremal terms in equation \eqref{bolide} do not explode.
By symmetry we can concentrate on $\int_{l(t)}^{i_1(t)} k \log k \dd x$.
Similarly to \eqref{bypartitis}, we have
\begin{equation}
\partial_t\left(\int^{i_1(t)}_{l(t)} k \log k \dd x \right)=k^2(l(t),t)-\int^{i_1(t)}_{l(t)} \frac{k^2_{x}}{k}\dd x.
\end{equation}

Note that $\int_{l(t)}^{i_1(t)} k(l(t),t)\dd x=1-f_x(i_1(t))$ is decreasing, because $f_x(i_1(t))$ is a local minimum (thus increases).
Suppose that 
$$\lim_{t\to t_1} \int_{l(t)}^{i_1(t)} k \dd x< 1.$$
Then for $t$ close to $t_1$ can use Lemma \ref{agmon} with
\begin{equation}
\gamma(x)=
 \begin{cases}
\sqrt{k(x+l(t),t)} &\text{ if }  x\le i_1(t)-l(t),\\
0 &\text{ if }  x\ge i_1(t)-l(t).
\end{cases}
\end{equation}
and obtain that
\begin{equation}\label{arcogive}
\int^{i_1(t)}_{l(t)} \frac{k^2_{x}}{k}\dd x< k^2(l(t),t).
\end{equation}
and hence that 
 $\int^{i_1(t)}_{l(t)} k \log k$ is decreasing in a neighborhood of $t_1$.
 
 \medskip
 
 Hence $\int^{i_1(t)}_{l(t)} k \log k$ can only explode if 
 \begin{equation}\label{arcoussin}
  \lim_{t\to t_1} \int_{l(t)}^{i_1(t)} k\dd x\ge 1.
\end{equation}
In \eqref{arcoussin} holds, let us consider $i_0(t)< i_1(t)$ such that
$$\int_{l(t)}^{i_0(t)} k(l(t),t)=1.$$ 
The point $i_0(t)$ is the point at which the leftest local maximum of $f$ is attained.
Let us first show that if $\int^{i_1(t)}_{l(t)} k \log k \dd x$ is unbounded, then in a neighborhood of $t_1$
$f$ must reach its maximum at $i_0(t)$ and $f(i_0,t)$ is small.

\medskip

Let us define $a(K,t)$ like in \eqref{ak}. Then if $\int_{l(t)}^{a(K,t)} k \dd x\le 1-\gd$
then the computation of the previous section \eqref{reptile} to \eqref{croconodile} are still valid and thus
\begin{equation}
 \partial_t\left(\int^{i_1(t)}_{l(t)} k \log k\right)\le \frac{256 K^4}{\delta^4}.
\end{equation}
This implies that if $\int^{i_1(t)}_{l(t)} k \log k \dd x$ is unbounded then for every $K$ and $\delta$,
 there is some $t\le t_1$ such that 

\begin{equation}\label{chfleur}
\int_{l(t)}^{a(K,t)} k(x,t) \dd x\ge 1-\gd.
\end{equation}
If \eqref{chfleur} holds, then Lemma \ref{akk} and the concavity of $f$ restricted to $[0,i_1(t)]$ imply that
\begin{equation}
 f(i_0(t),t)\le \delta(i_0(t)-a(K,t))+f(a(K,t),t)\le K^{-2}+\delta(r_0-l_0).
\end{equation}
If $f(i_0(t),t)$ is not the overall maximum of $f$ then it means that $f$ admits a local minimum which is smaller than
$K^{-2}+\delta(r_0-l_0)$. 
As local minima of $f$ are increasing, 
they must all be higher than local minima of $f_0$ at all time, and thus cannot be smaller than 
$K^{-2}+\delta(r_0-l_0)$, if $K$ is large enough and $\delta$ small enough.
Let us write the conclusion of this reasoning as a lemma.

\begin{lemma}
 If
 $$\limsup_{t\to t_1} \int^{i_1(t)}_{l(t)} k \log k\dd x= \infty,$$
 then in a neighborhood of $t_1$, $f(\cdot,t)$ has its maximum at
 $i_0\in (l(t),i_1(t))$, and $f$ has no other local extremum.
 Furthermore 
 \begin{equation}\label{minimax}
 \lim_{t\to t_1} f(i_0(t),t)=0.
 \end{equation}
\end{lemma}

 Our last task is to show that this is impossible, and thus that $f$ should become convex before $t_1$.
 
 As  for $t$ sufficiently large $i_0(t)$ is the only local maximum of $f$
\begin{equation}\label{pti}
 \lim_{t\to t_1} \int_{i_{2p}(t)}^{r(t)} k \dd x= \lim_{t\to t_1} (1+f_x(i_{2p}(t),t))<1.
\end{equation}
Indeed for sufficiently large $t$, $f_x(i_{2p}(t),t)\le 0$ because there is no local maximum of $f$ in $[i_{2p}(t),r(t)]$ ; and $f_x(i_{2p}(t),t)$ is a local maximum of $f_x$ and thus decreases strictly in time.

\medskip

Note also that 
\begin{equation}
 \int_{i_{2p}(t)}^{r(t)}k \dd x >  \int_{i_{2p-1}(t)}^{r(t)}k \dd x=1+f_x(i_{2p-1}(t),t).
\end{equation}
The r.h.s.\ in the above equation is positive and
 strictly increasing in $t$, because  $f_x(i_{2p-1}(t),t)$ is a local maximum of $f_x$
 
 \begin{equation}\label{gra}
 \lim_{t\to t_1} \int_{i_{2p}(t)}^{r(t)} k \dd x= \alpha\in(0,1).
\end{equation}

For $x\in(i_{2p}(t),r(t))$,
$f_x\in [-1,-1+\alpha]$, and thus
$$f(i_{2p}(t),t)\ge (1-\alpha)(r(t)-i_{2p}(t)).$$
As $f(i_{2p}(t),t)\le f(i_{0}(t),t)$, equation \eqref{minimax} implies that 
$$\lim_{t \to t_1} r(t)-i_{2p}(t)=0.$$ 
Thus the mean value of $k$ on the interval $[i_{2p}(t),r(t)]$ explodes:
\begin{equation}
\lim_{t\to t_1} \bar k(t):= \lim_{t\to t_1} \frac{\int_{i_{2p}(t)}^{r(t)} k\dd x }{(r(t)-i_{2p}(t))}=\infty.
\end{equation}
By Jensen's inequality, we have
\begin{equation}
 \int_{i_{2p}(t)}^{r(t)} k\log k \dd x \ge (r(t)-i_{2p}(t))(\bar k\log \bar k)(t)\ge \alpha \log \bar k (t),
\end{equation}
and thus $$\lim_{t\to t_1}\int_{i_{2p}(t)}^{r(t)} k\log k \dd x=\infty.$$

However, with the same argument used to obtain Equation \eqref{arcogive}, the inequality \eqref{pti} implies that
$$\int_{i_{2p}(t)}^{r(t)} k\log k \dd x \quad \text{ is uniformly bounded,}$$
yielding a contradiction.
\qed

\section{Preliminaries}\label{prelim}

\subsection{Stochastic domination and monotonicity in $\gl$/boundary condition}\label{monoton}

Our dynamics has quite enjoyable monotonicity properties that can be proved by standard coupling argument using the so-called {\sl graphical construction}.
First introduce a natural order on $\gO_0^L$. Given for two elements $\xi$ and $\xi'$ we say that $\xi\ge \xi'$ if $\xi_x\ge \xi'_x$ for every $x\in[-L,L]$.
We say that a dynamic $\eta$ dominates stochastically $\eta'$ if one can couple the two dynamic on the same probability space and have with probability one

$$\eta(\cdot,t)\ge \eta'(\cdot,t), \quad  \forall t>0$$
\medskip

We give some examples of monotonicity that we may use in what follows:
\begin{itemize}
 \item The dynamic with a wall and $\gl=1$ dominates the one without wall.
 \item If $\gl<\gl'$ the dynamic with parameter $\gl$ dominates the one with parameter $\gl'$.
\end{itemize}

\medskip

For the construction of the  coupling, we refer to  \cite[Section 2.1.1]{cf:CMT} where these things are very well
explained.

\subsection{A general upper-bound}

Using monotonicity, we prove here that the solution of \eqref{ssatef} is a general upper-bound for the scaling limit.
This provides half of Theorem \ref{repcase}, and will be of use for the proof of Theorem \ref{mainres}.
Here and in what follows we say that an event $A_L$ (or more properly, a sequence of event) occurs \textsl{with high probability} (we may also write w.h.p.) 
if the probability of $A_L$ tends to one when $L$ tends to infinity. 

\begin{proposition}\label{coucc}
 For all choices of $\gl\in [0,\infty]$, the dynamic starting with initial condition satisfying
\begin{equation} \label{arfs}
\eta^L_0(x)=Lf_0(x/L)(1+o(1)) \text{ uniformly in $x$ when $L\to \infty$ },
\end{equation}
is such that for any given $\gep>0$ and $T>0$, w.h.p
\begin{equation}
\frac{1}{L}\eta^L(Lx,L^2t)\le \tilde f(x,t)+\gep, \forall x\in[-1,1],t\in [0,T].
\end{equation}

\end{proposition}

\begin{proof}
We construct an alternative dynamics $\hat \eta$ that constitutes an upper bound for $\eta$.
 The dynamics $\hat \eta$ has the same transition rates as $\eta$ except that transition $\eta\to \eta^{(x)}$ for $x$ at a distance smaller than $L^{3/4}$ from the boundary are rejected. We also modify the initial condition slightly so that
\begin{itemize}
\item $\hat \eta^L_0(x)=x+L$, for $x\in [-L,-L+2L^{3/4}]$, $\hat \eta^L_0(x)=L-x$, for $x\in [L-2L^{3/4},L]$,
\item $\hat \eta^L_0\ge \eta^L_0$ and $\hat \eta^L_0(x)\ge 2L^{3/4}$, $\forall x\in [-L+2L^{3/4},L-2L^{3/4}]$,
\item $\hat \eta^L_0$ satisfies \eqref{arfs}.
\end{itemize}

From its initial condition and constraint it follows that $\hat \eta$ is an upper bound for $\eta$.
Moreover, up to the first time of contact with the wall, in $[-L+L^{3/4},L-L^{3/4}]$, $\hat \eta$ 
coincides with a corner-flip dynamics, and as seem in Lemma \ref{mousis} in the appendix, with large probability $\hat \eta$ does 
not touch the wall before time $L^2T$. 
Thus we can apply Theorem \ref{dynsanw} to the corner-flip dynamics on the segment $[-L+L^{3/4},L-L^{3/4}]$ to get the result.

\end{proof}

\section{Proof of Theorem \ref{repcase}}\label{repphase}

\subsection{The case $\gl\in[0,1]$}

The Proposition \ref{coucc} already provides the upper-bound part of the convergence, so what remains to do is to prove that 
for every $\gep$ with high probability
$$ \forall x\in [-1,1], \forall t\in [0,T], \quad  \frac{1}{L}\eta^L(Lx,L^2t)\ge \tilde f(x,t)-\gep.$$
From Theorem \ref{frefer} the above inequality is satisfied, when $\eta$ is replaced by the dynamics without wall $\tilde \eta$. 
Moreover from Section \ref{monoton}, one can couple the two dynamics such that $\tilde \eta(t) \le \eta(t)$ for all $t$, and hence the result follows.

\subsection{The case $\gl\in(1,2)$}

When $\gl>1$, there is no simple stochastic comparison available and one must work harder to obtain the result.
The idea we use  (which is also present in \cite[Section 5]{cf:CMT}, but used to prove bounds on the mixing time), is that when $\gl\le 2$
the function $\Phi$ defined on the set of paths $\gO$ as
\begin{equation}
\Phi(\eta):=\sum_{x=-L}^L g(x)\eta(x),
\end{equation}
where
$$g(x):= \cos\left(\frac{x\pi}{2L}\right),$$ 
is close to be an eigenfunction of the generator of our Markov chain.

\medskip

In the remainder of the paper, for notional convenience we write $\eta(t)$ for $\eta(\cdot, t)$.

\begin{proposition}\label{downar}
When $L$ is large enough,
for any $\eta_0\in \gO_L$ for all $t\le L^{2+\gep}$

\begin{equation}\label{enproba}
|\bbE[\Phi(\eta(t))]-
\exp\left(-t\pi^2/(2L)^2\right)\Phi(\eta_0)|\le L^{7/4}.
\end{equation}
Furthermore, if $\eta_0$ satisfies \eqref{assum} for a given $f_0$, then
\begin{equation}\label{gos}
\lim_{L\to \infty}\max_{t\in[0,T]}\left|\frac{1}{L^2}\Phi(\eta(L^2t))-\exp\left(-t\pi^2/4\right)\int_{-1}^1 f_0(x)\cos(\pi x/2)\dd x\right|=0.
\end{equation}
\end{proposition}

A way to reformulate \eqref{gos} is that the Fourier coefficient of the rescaled interface converges to the one of $\tilde f$.
We define $\bar \eta$ to be the rescaled version (defined on $[-1,1]$)
\begin{equation}\label{barete}
 \bar \eta(x,t)=\frac{1}{L}\eta(Lx,L^2t).
\end{equation}
Then \eqref{gos} can be read as
\begin{equation}\label{fourcoef}
\lim_{L\to\infty} \int_{-1}^1 \bar \eta(x,t)\cos(\pi x/2)\dd x=\int_{-1}^1 \tilde f(x,t) \cos(\pi x/2)\dd x,
\end{equation}
where convergence holds uniformly in $[0,T]$, in probability.
As Proposition \ref{coucc} already provides one bound, this estimate turns out to be sufficient to prove convergence of $\eta$.

\begin{proof}[Proof of Theorem \ref{repcase} for $\gl \in (1,2)$ from Proposition \ref{downar}]

Using the fact that $|y|= 2y_+ -y$  (where $y_+=\max(y,0)$) we have
\begin{multline}
\int_{-1}^1|\bar \eta(x,t)-f(x,t)|\cos(\pi x/2)\dd x=\\2 \bbE\left[\int_{-1}^1(\bar \eta(x,t)-f(x,t))_+\cos(\pi x/2)\dd x\right]
+ \int_{-1}^1(f(x,t)- \bar \eta(x,t))\cos(\pi x/2)\dd x.
\end{multline}

Using Proposition \ref{coucc}, we know that the first term tends to zero uniformly on $[0,T]$ in probability.
This is also the case of for the second term thanks to \eqref{gos} or \eqref{fourcoef}.

\medskip

Hence for any $\gep>0$, w.h.p.\ for all $t\le T$
\begin{equation}\label{crocodileum}
\int_{-1}^1|\bar \eta(x,t)-f(x,t)|\cos(\pi x/2)\dd x\le \gep.
\end{equation}
To conclude, we use the fact $|\bar \eta(x,t)-f(x,t)|$ is a $2$-Lipshitz function to show that \eqref{crocodileum} implies 
uniform convergence.

\medskip

Note that 
$$\forall x\in [-1,-1-\delta]\cup[1-\delta,1],\quad  |\bar \eta(x,t)-f(x,t)|\le \delta$$
because both $f$ and $\bar \eta$ are in $[0,\delta]$.
To control $|\bar \eta(x,t)-f(x,t)|$ on $[-1+\delta]\cup[1-\delta]$ we notice that \eqref{crocodileum} implies that
\begin{equation}
\int_{-1+\delta}^{1-\delta} |\bar \eta(x,t)-f(x,t)|\dd x\le \frac{\gep}{\sin(\pi\delta/2)}\le \gep/\delta.
\end{equation}
As $\bar \eta-f$ is a $2$-Lipshitz function in $x$, this implies that whenever \eqref{crocodileum} holds,
\begin{equation}
 |\bar \eta(x,t)-f(x,t)|\le \sqrt{\gep/\delta},\quad  \forall x \in [-1+\delta,1-\delta].
\end{equation}
Taking choosing $\gep=\delta^3$ we conclude that w.h.p.\ for all $t\in[0,T]$
$$ |\bar \eta(x,t)-f(x,t)|\le \delta.$$

\end{proof}

\begin{proof}[Proof of Proposition \ref{downar}]
Using \cite[Lemma 2.3, (2.37)]{cf:CMT} (be careful that the definition for the discrete Laplacian differs by a factor $2$ and 
the same apply to our transition rates), and the linearity of $\mathcal L$ of  we have

\begin{multline}\label{fifi}
\partial_t \bbE\left[\Phi(\eta(t))\right]= \bbE\left[(\mathcal L\Phi)(\eta(t))\right]=\\
\bbE\left[\sum_{x=-L}^L g(x)(\gD \eta)(x,t)\right]
+2\bbE\left[\sum_{x=-L}^L  g(x) \ind_{\eta(x\pm 1,t)=0}\right]\\-\frac{2(\gl-1)}{\gl+1}\bbE\left[\sum_{x=-L}^L  g(x) \ind_{\eta(x\pm 1,t)=1}\right].
\end{multline}
Doing summation by part and using the fact that $\gD g(x)=-\kappa_L g(x)$ where $$\kappa_L:=2\left(1-\cos\left(\pi/2L\right)\right))$$
the first term is equal to $-\kappa_L\bbE[\Phi(\eta(t)]$. We can control the value of the two other terms with the following estimates:
\bigskip

\begin{lemma}\label{contact}
For any $\delta>0$,
there exists a constant $C_1(\gl,\delta)$ such that
for any choice of initial configuration $\eta_0\in \gO^0_L$, one has for every $t\ge 0$, and every $x$
\begin{equation}\label{contactor}\begin{split}
\bbP[\eta(x\pm 1,t)=1]&\le C_1 \min(d_L(x), t^{1/2-\delta})^{-3/2} \text{ if $x+L$ is even},\\
\bbP[\eta(x\pm 1,t)=0]&\le C_1 \min(d_L(x), t^{1/2-\delta})^{-3/2} \text{ if $x+L$ is odd},
\end{split}\end{equation}
where $d_L(x)=\min(|x+L|,|x-L|)$ denotes the distance between $x$ and the boundary of $[-L,L]$.
\end{lemma}

We postpone the proof to the end of the Section.
Note that Lemma \ref{contact} implies that for every $t\le L^{\frac{2}{1-2\delta}}$
\begin{equation}\label{bibero}\begin{split}
\bbE\left[\sum_{x=-L}^L  g(x) \ind_{\eta(x\pm 1,t)=1}\right]\le C_2 L t^{-3/4+3\delta/2},\\
\bbE\left[\sum_{x=-L}^L  g(x) \ind_{\eta(x\pm 1,t)=0}\right]\le C_2 L t^{-3/4+3\delta/2}. 
\end{split}\end{equation} 
and thus
 \begin{equation}
|\partial_t \bbE\left[\Phi(\eta(t))\right]+\kappa_L\bbE[\Phi(\eta(t)]|\le C_3 L t^{-3/4+3\delta/2}.
 \end{equation}
Integrating the above inequality, we obtain for all $t\le L^{\frac{2}{1-2\delta}}$

\begin{multline}
 \bbE\left[\Phi(\eta(t))\right]\ge \exp(-t\kappa_L)\Phi(\eta_0)-C_3\int_0^t e^{\kappa_L(s-t)} L t^{-3/4+3\delta/2}\dd s\\
 \ge \exp(-t\kappa_L)\Phi(\eta_0)-C_3L t^{1/4+3\delta/2}\ge  \exp(-t\kappa_L)\Phi(\eta_0)-L^{7/4},
\end{multline}
if $\delta$ is small enough. A lower-bound can be obtained in the same manner accordingly.
And hence \eqref{enproba} holds.
\medskip

Let us now use the convergence of $\bbE\left[\Phi(\eta(t))\right]$ to prove \eqref{gos}. As we already have the upper-bound
which is just a consequence of 
Proposition \ref{coucc}, we only need to prove that for any $\delta$ w.h.p.\
\begin{equation}\label{gosp}
\forall t\in[0,T],\quad \frac{1}{L^2}\Phi(\eta(L^2t))-\exp\left(-t\pi^2/4\right)\int_{-1}^1 f_0(x)\cos(\pi x/2)\dd x\ge -\delta.
\end{equation}

\medskip
Let us fix $T$ and $\delta>0$ arbitrary  $\gep=\gd^2/4$.
Fom Proposition \ref{coucc}, we have for $L$ large enough, for any event $A$
\begin{equation}\label{lesa}
\bbE\left[\Phi(\eta(L^2T))\ind_A \right]\le (\bbP[A]+\gep)\exp(-\pi^2 T/4)\Phi(\eta_0).
\end{equation}
Let us define
\begin{equation}
\tau:= \min \{ t\  | \ \Phi(\eta(L^2t))\le (1-\delta)\exp(-\pi^2t/4)\Phi(\eta_0)\}.
\end{equation}
As $\tau$ is a stopping time, we can apply the Markov property and \eqref{enproba} for the initial condition $\eta_{\tau}$ to obtain that 
for every $t\le T$.

\begin{multline}
 \bbE\left[ \Phi(\eta(L^2(\tau+t)))\ | \ (\tau,\eta_{\tau}) \right]\le \exp(-\pi^2 t/4)\Phi(\eta_{\tau})+L^{7/4}\\
 \le (1-\delta)\exp(-\pi^2 (\tau+t)/4)\Phi(\eta_0)+L^{7/4}.
\end{multline}

This inequality remains also valid if $t$ is a function of $\tau$. Taking $t=T-\tau$ we have, on the event 
$\{\tau\le T\}$,
\begin{multline}
 \bbE\left[ \Phi(\eta(L^2T)) | \ \tau \right]\le (1-\delta)\exp(-\pi^2 T/4)\Phi(\eta_0)+L^{7/4}\\
 \le 
 (1-\delta/2)\exp(-\pi^2 T/4)\Phi(\eta_0).
\end{multline}
Combining it with \eqref{lesa} for the event $\{\tau>T\}$ we have 
\begin{multline}
  \bbE\left[ \Phi(\eta(L^2T))\right]\le (1-\delta/2)\exp(-\pi^2 T/4)\Phi(\eta_0)\bbP[\tau\le T]
  + \bbE\left[ \Phi(\eta(L^2T))\ind_{\tau>T}\right]\\
  \le  (1-\delta/2\bbP[\tau\le T]+\gep)\exp(-\pi^2 T/4)\Phi(\eta_0).
\end{multline}
On the other hand \eqref{enproba} implies that
\begin{equation}
   \bbE\left[ \Phi(\eta(L^2T))\right]\ge  (1-\gep)\exp(-\pi^2 T/4)\Phi(\eta_0),
\end{equation}
and hence as $\gep=\gd^2/4$,
\begin{equation}
 \bbP[\tau\le T]\le 4\gep/\delta\le \delta.
\end{equation}
 Hence with probability larger than $1-\delta$
\begin{equation}
 \Phi(\eta(L^2t))\ge \exp(-\pi^2 t/4)\Phi(\eta_0)(1-\delta) \quad \forall t\in [0,T].
\end{equation}
Dividing by $L^2$ leads to \eqref{gosp}.

\end{proof}

\begin{proof}[Proof of Lemma \ref{contact}]
We can assume $t\ge C^{4/(3+2\delta}_1)$, as if not the result for holds trivially.
By monotonicity, it is sufficient to prove the result with the smallest possible initial condition, where $\eta_0=\eta^{\min}$ (recall \eqref{etamin})
\medskip

Let us treat the case $x+L$ even only as the second line of \eqref{contactor} can be proved in a similar manner.
 Consider first $x$ which is at a distance larger than $t^{1/2-\delta}$ from the boundary. To avoid complicated notation we assume that $t^{1/2-\delta}$ is an even integer (or else, we replace it by twice the integer part of its half). By mononicity again, the dynamics $(\eta(\cdot,s))_{s\ge 0}$ can be coupled with a dynamics 
 $(\eta'(\cdot,s)_{s\ge 0}$ with a wall and pinning force $\gl$ on state space
 \begin{multline}
  \gO_{x,t}:=\{ \eta\in\bbZ^{2t^{1/2-\delta}+1}\ | \;
\eta_{x\pm t^{1/2-\delta}}=0\,; \\
 \forall y\in\{x-t^{1/2-\delta},\dots,x+t^{1/2-\delta}-1\}, |\eta_{x+1}-\eta_x|= 1, \eta_x\ge 0 \}\, .
 \end{multline}
We can perform the coupling in such a manner that 
$$ \forall y\in\{x-t^{1/2-\delta},\dots,x+t^{1/2-\delta}\},\  \forall s\ge 0, \eta(y,t)\ge \eta' (y,t) $$
\medskip

Hence it is sufficient to prove the result for $\eta'$.
From \cite[Theorem 3.1]{cf:CMT}, we know that the mixing time $\Tm(1/2e)$ of such a dynamics is smaller than $t^{1-\delta}$.
Hence at time $s=t$, the total variation distance between the distribution of $\eta'(\cdot,t)$ and the equilibrium distribution $\pi':=\pi'_{x,t}$ satisfies

\begin{equation}\label{TVbound}
\| \bbP[\eta'(\cdot,t) \in \cdot] -\pi'\|_{TV}\le\exp(\lfloor t/ \Tm(1/2e)\rfloor )\le \exp(-t^{\delta}/2).
\end{equation}
At equilibrium, the probability that the midpoint of a polymer of length  $2t^{1/2-\delta}$ is pinned satisfies 
\begin{equation}\label{crogs}
\pi'( \eta'(x\pm 1)=1)=\frac{1+\gl}{\gl} \pi'(\bar \eta(x)=0)
=\frac{1+\gl}{\gl}\frac{(Z^{\gl}_{t^{1/2-\delta}})^2}{(Z^{\gl}_{2t^{1/2-\delta}})^2}\le C_\gl t^{-3/2+3\delta}.
\end{equation}
In the last inequality, we used the asymptotic equivalence 
$$Z^{\gl}_{2l}\approx c_\gl l^{-3/2}$$ which holds for an explicit constant $c_\gl$, see \cite[Theorem 2.2 (2)]{cf:GB}.
Hence from \eqref{TVbound} one has
\begin{equation}
\bbP[\eta'(x,t)=0]\le  C_\gl t^{-3/2+3\delta}+\exp(-t^{\delta})\le C_1  t^{-3/2+3\delta}.
\end{equation}
if $C_1$ is chosen in an appropriate manner.

\medskip

When $x$ is at a distance smaller than $t^{1/2-\delta}$ from the boundary (say from $-L$), we use the same idea and compare the dynamics to one on the restricted path space  
\begin{multline}
  \gO_{x}:=\{ \eta\in\bbZ^{2(L+x)+1}\ | \;
\eta_{-L}=0, \eta_{2x+L}=0\,; \\
 \forall y\in\{-L,\dots,2x+L\}, |\eta_{x+1}-\eta_x|= 1, \eta_x\ge 0 \}\, .
 \end{multline}
From \cite[Theorem 3.1]{cf:CMT}, this restricted $\eta'$ on an interval of length $2(L+x)$, for which the mixing time is smaller than $t^{1-\delta}$
and thus similarly
\begin{equation}
\bbP[\eta'(x,t)=0]\le \pi'( \eta'(x\pm 1)=1) +\exp(-t^\delta)\le  2C(L+x)^{-3/2}.
\end{equation}
where $\pi'$ denote the equilibrium measure for the restricted dynamics. 
\end{proof}

\section{Proof of Theorem \ref{mainres}} \label{mainressec}

Let us describe a bit the strategy behind the proof of Theorem \ref{mainres}.
The first step of the proof,   Lemma \ref{drifta}, is to show that the area below $\eta$ decays at least with a constant rate equal to $-2$ 
on the rescaled picture. This is performed by computing the expected drift of the area, and using a
martingale technique to bound the possible fluctuation.

\medskip

Using this bound on the area, we show that the problem can be reduced to proving that the solution of \eqref{sstef} is asymptotically 
a lower bound for the evolution 
of the polymer (see Proposition \ref{lb} and below).
The most difficult part then, when proving Proposition \ref{lb} 
is to control the motion of the boundary of the unpinned zone for $\eta(t)$.

%
%
%
%
%

\subsection{An upper-bound for the decay of the area below the graph of $\bar \eta$}

Consider $$a(\bar \eta(t))=\int_{-1}^1 \bar\eta(x,t)\dd x$$ 
the area below the rescaled polymer (recall \eqref{barete})
(recalle the notation $\eta(t)=\eta(\cdot,t)$). 
Similarly let $a(f(t))=\int_{-1}^1 f(x,t)\dd x$  denote the area below the graph of $f(\cdot,t)$.
We have seen (recall \eqref{area}) that 
\begin{equation}\label{area2}
a(f(t))=a(f(0))-2t \text{ when } t\le T^*.
\end{equation}
We want to prove a similar statement for $a(\bar \eta(t))$.
\medskip

To state our result, it is easier to consider the area below the non-rescaled curve:
\begin{equation}
 A(\eta(t)):=\int_{-L}^L \max(\eta(x,t),1)\dd x.
\end{equation}
Note that $A$ is not exactly the area because of the $\max$ present in the integral, but this detail is of no importance once 
rescaling is performed. The quantity $\max(\eta(x,t),1)$ is present instead of $\eta(x,t)$ in order to have a nice expression for the expected drift of $A$
(see below).

It can be checked by the reader (see also Section 5.4 in \cite{cf:CMT}) that the expected drift of $A(\eta(t))$ 
is equal to (recall \eqref{generator}
\begin{multline}\label{drift}
 (\mathcal L A)(\eta)= 2 
|\{ x\in  \{-L\dots L\} \ | \ \eta_x=\eta_{x\pm 1}-1  \text{ and } \eta_{x\pm 1} \ne 0\}|\\
-2|\{ x\in  \{-L\dots L\} \ | \ \eta_x=\eta_{x\pm 1}+1 \text{ and }
 \eta_x\ne 1 \}|.
\end{multline}
This is because the transition $\eta\to \eta^{(x)}$ increases/decreases $A$ by $\pm2$, 
except if it adds one contact with the wall, in which case $A$ is 
decreased by $-1$ (but this happens with a rate twice as big).

\medskip

The right-hand side of \eqref{drift} is equal to minus $2$ times the number of excursions of length $4$ or more away from the wall, and thus 
\begin{equation}\label{groove}
  (\mathcal L A)(\eta)\ge -2, \forall \eta \in \Omega_L\setminus \{ \eta_{\min}\}.
\end{equation}

\medskip

\begin{lemma}\label{drifta}
One has w.h.p., 
\begin{equation}\label{fdf}
 A(\eta(t))\le A(\eta(0))+\int_0^{t} \mathcal L A(\eta(s))\dd s+ L^{7/4}, \quad  \forall t\in[0,L^2],
\end{equation}
so that in particular for all $t\in[0,L^2]$
\begin{equation}\label{asdf}
 A(\eta(t))\le \max(A(\eta(t))-2t+L^{7/4}, 2L).
\end{equation}
As a consequence, w.h.p.\ 
\begin{equation}\label{croci}
 \mathcal T\le A(\eta(0))/2+L^{7/4},
\end{equation}
and w.h.p.\ uniformly for all time,
\begin{equation}\label{trucmachin}
 a(\bar \eta(t))\le (a(f_0)-2t)_+ + 2L^{-1/4}.
\end{equation}

\end{lemma}
\begin{proof}

It is a standard property of Markov chains that 
$$M_t:=A(\eta(t))- A(\eta(0))-\int_0^t (\mathcal L A (\eta(s)))\dd s$$ is a martingale, for the filtration associated with the process $(\eta(t),\ t\ge 0)$.
To prove \eqref{fdf} we have to show that $M_t$ cannot be too large. We use Doob's maximal inequality
\begin{equation} \label{doob}
\bbP[\max_{t\in [0,L^2]} M_t\ge L^{7/4}]\le L^{-7/2} \bbE[M^2_{L^2}].
\end{equation}
As $M_0=0$ the expected value of $M^2_{L^2}$ is the one of the martingale bracket $\langle M^2\rangle_{L^2}$, for which we have an explicit expression

\begin{equation}
 \langle M^2\rangle_{t}=\int_0^t F(\eta(s))\dd s
\end{equation}
where
\begin{multline}\label{vardrift}
 F(\eta(s)):= 4 |\{ x\in  \{-L\dots L\} \ | \ \eta_{x+1}=\eta_{x-1}\notin \{0,1\}\}|\\
 +  2 |\{ x\in  \{-L\dots L\} \ | \ \eta_x=2,\ \eta_{x\pm 1}= 1\}| \le 8L.
\end{multline}
Hence from \eqref{doob}
\begin{equation}
 \bbP[\max_{t\in [0,L^2]} M_t\ge L^{7/4}]\le 8 L^{-1/2}.
\end{equation}

For the second statement, we note that
either the chain has already reached $\eta_{\min}$ at time $t$ and thus $A(t)=2L$ 
or it has not and from \eqref{groove}
$$\int^t_0 \mathcal L A(\eta(s))\dd s \ge -2t,$$ so that \eqref{asdf} is a consequence \eqref{fdf}.

\medskip
Finally, to obtain \eqref{croci} we notice that if $\mathcal T\ge A(\eta(0))/2+L^{7/4}$, then
\begin{equation}
 \int_0^{A(\eta(0))+L^{7/4}} \mathcal L A(\eta(s))\dd s\le -2A(\eta(0))-2L^{7/4}.
\end{equation}
Thus this cannot occur with non-vanishing probability, as it would bring a contradiction to \eqref{fdf}.
Equation \eqref{trucmachin} is obtained by rescaling time and space in \eqref{asdf} (we a have to replace $A$ by the area below the curve, but the difference between the 
two is at most $L$).
\end{proof}

\subsection{Reducing the problem to the proof of an asymptotic lower-bound for $\eta(t)$}

A consequence of equation \eqref{trucmachin} is that in order to prove Theorem \ref{mainres}, we only need a lower-bound result.
Indeed the upper-bound on the area plus the constraint that $\bar \eta(\cdot, t)$ is a Lipshitz function are sufficient to deduce the upper-bound from the lower-bound.
For this reason, in the remainder of the paper we focus on proving

\begin{proposition}\label{lb}
Let $\eta^{L,\infty}$ be the dynamic with wall and $\gl=\infty$, 
and starting from a sequence of initial condition
$\eta^L_0$ satisfying 
\begin{equation} \label{caref}
\eta^L_0(x)= Lf_0(x/L)(1+o(1)) \text{ uniformly in $x$ when $L\to \infty$ }.
\end{equation}
Then for every choice of $\gep>0$, the rescaled dynamics  $\bar \eta$ satisfies w.h.p.
\begin{equation}
\bar \eta^{L,\infty}(x.t)\ge f(x,t)-\gep, \quad \forall x\in[-1,1], \forall t>0.
\end{equation}
\end{proposition}

\begin{proof}[Proof of Theorem \ref{mainres} from Proposition \ref{lb}]
It is sufficient to prove that for any $\delta>0$,  w.h.p
\begin{equation}
  \forall x\in[-1,1], \forall t>0 \quad \bar \eta(x.t)\le f(x,t)-\delta.
  \end{equation}
Combining Equations \eqref{trucmachin} and \eqref{area} we have w.h.p
\begin{equation}
   \forall x\in[-1,1], \forall t>0,\quad a(\bar \eta(t))\le a(f(\cdot,t))+\gd^2/32.
\end{equation}
Moreover, from Proposition \ref{lb} we have w.h.p.\ for all $x$ and $t$
\begin{equation}
 \bar \eta(x,t)\ge f(x,t)-\gd^2/64.
\end{equation}
Hence w.h.p.\  for all $t>0$ 
$$x\mapsto\bar \eta(x,t)- f(x,t)+\gd^2/32$$ is a $2$-Lifshitz positive function whose integral is smaller than
$\gd^2/16$. This implies that
\begin{equation}
 \bar \eta(x,t)- f(x,t)-\gd^2/32 \le \gd/2.
\end{equation}
Concerning \eqref{ouichtrucfacile}, the upper-bound on $\mathcal T$ is proved in Lemma \ref{drifta}, and the lower-bound is a consequence 
of \eqref{trouille}.

\end{proof}

\subsection{Overall strategy for the proof of Proposition \ref{lb}}

The main difficuly when trying to prove Proposition \ref{lb} is to control the motion of the boundary between the pinned and the unpinned zone.
Indeed, for part of $\eta$ that are far from the wall, Theorem \ref{frefer} can be used to control the drift of the rescaled polymer.
The first and most novel idea in the proof is to add a a small perturbation of amplitude $\delta$ to the function $f$ and 
to the initial condition $\eta_0^L$ 
(see the caption of Figure \ref{piedestal}). The reason for adding this perturbation is that when adding flat parts of slope $+1/-1$ on the sides of $\eta$,
the motion of the phase boundary i.e.\ the boundary between the pinned and unpinned phases, is slowed down.
Of course a consequence of this modification is that the initial condition does not satisfy \eqref{caref}, and we have
to find an equivalent formulation Proposition \ref{lb} that deals with this problem: this is Proposition \ref{infdelta}.

\medskip

We show that it is enough to control where the dynamics goes during a time period $\gep$ and iterate the process.
This is the role of Lemmata \ref{iltep} and \ref{iltep2}. 

\medskip

%
%

\subsection{Modification of the initial function}

Given $f_0$ with $l_0>-1$, $r_0<-1$, set (recall \eqref{area2})

\begin{equation}
\bar \delta=a(f_0)\delta/2 \quad \text{and} \quad \bar\delta(t):=\delta a(f(t))/2=\delta(a(f_0)/2-t). 
\end{equation}

We define $f^{\delta}:[-1,1]\times (0,\infty)\to \bbR_+$ (for $\gd$ small enough)
by
\begin{equation}\label{clep}
 f^{\delta}(x,t):=\begin{cases}
    f(x,t)+\bar \delta(t), \text{ for }  x\in (l(t),r(t)),\\
    x-(l(t)-\bar \delta(t)), \text{ for } x\in  (l(t)-\bar\delta(t),l(t)),\\
    -\big(x-(r(t)+\bar\delta(t))\big), \text{ for } x\in  (r(t),r(t)+\bar\delta(t)),\\
   0 \quad \text{elsewhere}.
   \end{cases}
\end{equation}
See Figure \ref{piedestal} for a graphical vision of $f^{\delta}$. We also set 

$$f^\delta_0:=f^{\delta}(\cdot,0).$$

\begin{figure}[hlt]
 \begin{center}
 \leavevmode 
 \epsfxsize =14 cm
 \psfragscanon
 \psfrag{L}{\small $l(t)$}
 \psfrag{R}{\small $r(t)$}
\psfrag{Del}{\small $\bar\delta(t)$}
 \epsfbox{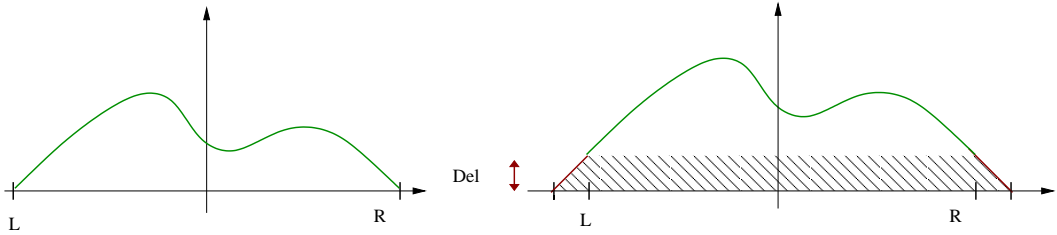}
 \end{center}
 \caption{\label{piedestal} The construction of the graph of $f^{\delta}(t)$ from the graph of $f(t)$ is done by adding some kind of pedestal of height
$\bar \delta(t)$ to support the original graph.}
 \end{figure}

From the fact that $f$ is the solution of \eqref{sstef}, we can deduce that $f^\delta$ satisfies
\begin{equation}\label{clep2}
\begin{cases}
  \partial_t f^\delta=f^{\delta}_{xx}-\delta \text{ for } x\in(r(t),l(t)),\\
  \partial_t f^\delta(x,t)=f_{xx}(l(t),t)-\delta \text{ for } x\in(l(t)-\bar\delta(t),l(t)),\\
  \partial_t f^\delta(x,t)=f_{xx}(r(t),t)-\delta \text{ for } x\in(r(t),r(t)+\bar \delta(t)).
   \end{cases}
\end{equation}

Instead of proving Proposition \ref{lb}, we prove a similar statement where $f$ is replaced by $f^\delta$.
Somewhere in the proof we will need some continuity assumption on the solution that are uniform for all time.
For this reason we will require our initial condition to be in the set

\begin{multline}
 \mathcal E:=\{ (f_0,r_0,l_0) \ |\ f_0 \text{ is positive, Lipshitz,} C^\infty \text{ on } (l_0,r_0), \\
  \text{ and such that } 
 l(t) \text{ and } r(t) \text{ are } C^\infty \text{ on } [0,T^*) \}.
\end{multline}

With this assumptions we are certain that there exists $C(f_0,c)$ that is such that
uniformly on $t\in[0,T^*-c/4]$ and $x\in l(t),r(t)$

\begin{equation}\label{lesbornes}\begin{split}
 |f_xx(x,t)|& \le C,\quad |(\partial_x)^3 f(x,t)|\le C,\\
|\partial_t f(x,t)|=|(\partial_x)^4 f(x,t)|\&le C,\\ 
\max(|r'(t)|,|r''(t)|,|l'(t)|,|l''(t)|)&\le C,\\
\end{split}\end{equation}

The set
$\mathcal E$ is dense for the uniform norm in the set of Lipshitz function that are positive on an interval.
This is because from Theorem \ref{existence}, the solution of \eqref{sstef} $f(\cdot,t)$ belongs to $\mathcal E$ for all positive times.

\begin{proposition}\label{infdelta}
Given a $f_0$ regular enough, and $\delta>0$
starting from a (sequence of) initial condition $\eta_0^L$ satisfying
\begin{equation}\label{crounch}
  \forall x\in[-1,1], \quad \bar \eta^L_0(x)\ge f_0^{\delta}(x),
\end{equation}

we have, for every $\gep>0$, w.h.p

\begin{equation}\label{corie}
\forall x\in[-1,1], \ \forall t>0,
 \quad \bar \eta(x,t)\ge f^{\delta}(x,t)-\gep .
 \end{equation}
\end{proposition}

The reason why Proposition \ref{infdelta} is easier to prove than Proposition \ref{lb} is that as can be seen from 
\eqref{clep2}, $f^{\delta}$ has a stronger ``push-down'' than $f$. Also even if $f^\delta_{xx}\equiv 0$ in a neighborhood $l(t)-\bar \delta(t)$ and $r(t)+\bar \delta(t)$, there is a contration of the boundary of the unpinned region
(meaning that the lateral push is also stronger).

\medskip

However, it is not too difficult to prove that the two statements are in fact equivalent.

\begin{proof}[Proof of Proposition \ref{lb} from Proposition \ref{infdelta}]

Given $f_0$, and $\gep$, we consider $\delta>0$ small enough (depending on $f_0$ and $\gep$) and 
we define 
$\hat f_0$ to be a $1$-Lipshitz function which is positive and smooth on $(\hat l_0, \hat r_0)\subset(l_0,r_0)$.
and that satisfies

\begin{equation}\label{troicinq}
 \hat f_0^{2\delta}\le f_0, \quad \text{ and } \quad
 \int_{-1}^1 (f_0-\hat f_0)(x)\dd x\le \gep^2/4.
\end{equation}
Let $\hat f$ denote the solution of \eqref{sstef} with initial condition $\hat f_0$ and $\hat f^{\delta}$ resp. $\hat f^\delta_0$, is defined as 
in \eqref{clep}, replacing $f$ by $\hat f$. We remark that if $\eta^L_0$ satisfies \eqref{caref} then for $L$ large enough
$$ \forall x\in[-1,1], \quad \bar \eta^L_0\ge \hat f^\delta_0.$$
Indeed it has to be checked only on the interval where $\hat f_0^{\delta}$ is positive, and on this interval 
$f_0- \hat f_0^\delta$ is uniformly bounded away from zero (from \eqref{troicinq}).

\medskip

Applying Proposition \ref{infdelta} to $\hat f_0$ we obtain that
w.h.p.\ 

\begin{equation}\label{andre}
\forall t\ge 0,\ \forall x\in[-1,1], \quad  \bar \eta(x,t)\ge \hat f^{\delta}(x,t)- \gep \ge \hat f(x,t)-\gep.
\end{equation}

\medskip

Applying Equation \eqref{area} to $f$ and $\hat f$ and combining it with the assumption \eqref{troicinq}, we obtain that for all $t>0$
\begin{equation}
 \int^1_{-1} (f(x,t)-\hat f(x,t))\dd x\le \gep^2/4.
\end{equation}
As $(f-\hat f)(\cdot,t)$ is $2-$Lipshitz and positive, this implies

$$\forall t>0,\ \forall x\in[-1,1], \quad  f \le \hat f+\gep.$$
Hence \eqref{andre} implies that
\begin{equation}
 \bar \eta(x,t)\ge f(x,t)-2\gep.
\end{equation}

\end{proof}

\subsection{Reduction to a statement for infinitesimal time}

To prove Proposition \eqref{infdelta}, we slice time into short period of length $\gep$ during which it is easier to control the dynamics.
The proof can be decomposed in two steps:
in Lemma \ref{iltep} we check that when $t$ is a multiple of $\gep$, the polymer stays above $f^\delta(\cdot,t)$.
Then Lemma \ref{iltep2} is used to fill the gap; it shows that during a period of time $\gep$ the polymer cannot go down too much.

\medskip

%

\begin{lemma}\label{iltep}
Given $f_0$ and $\delta>0$, and $c>0$, 
there exists $\gep_0=\gep_0(f_0,\delta,c)>0$ such that
for all $\gep\le \gep_0$, $k$ satifying $k\gep\le (a(f_0)-c)/2$,
if w.h.p.\
\begin{equation}\label{grominet}
, \forall x\in[-1,1], \quad  \eta(x,k\gep)\ge f^{\delta}(x,k\gep),
\end{equation}
then w.h.p.
\begin{equation}
, \forall x\in[-1,1], \quad \bar \eta(x,(k+1)\gep )> f^{\delta}(x,(k+1)\gep).
\end{equation}
\end{lemma}

\begin{lemma}\label{iltep2}
Given $f_0$, $\delta>0$, $c>0$, and $\eta>0$
there exists $\gep_1=\gep_1(f_0,\delta,c,\alpha)>0$ such that: 
for all $\gep\le \gep_1$ 
and $k$ satifying $k\gep\le (a(f_0)-c)/2$,
if w.h.p.\eqref{grominet} holds then
\begin{equation}\label{jkl}
  \forall t\in (\gep k,\gep(k+1)), \quad \bar \eta (x,t)> f^{\delta}(x,k\gep )-\alpha.
\end{equation}
\end{lemma}

\begin{proof}[Proof of Proposition \ref{infdelta} from Lemmata \ref{iltep} and \ref{iltep2}]
Given $f_0,\delta,c,\eta$, let us choose
$\gep=\min(\gep_0,\gep_1)$.
If $\eta^L_0$ satisfies \eqref{crounch} then reasoning by induction and using 
Lemma \ref{iltep} one obtains setting $k_{\max}=\lceil (a(f_0)-c)/2\gep \rceil$
that w.h.p.
\begin{equation}
 \forall k \in [0,k_{\max}], \quad \bar \eta(x,k\gep )> f^{\delta}(x,k\gep).
\end{equation}
Then we can use Lemma \ref{iltep2} to get a lower bound on $\bar \eta(x,t)$ for intermediate times and we obtain that 
w.h.p.\
\begin{equation}
  \forall t\in[0,(a(f_0)-c)/2], \quad \bar \eta(x,t)> f^{\delta}(x,t)-\alpha.
\end{equation}
Finally, for $t\ge (a(f_0)-c)/2$, Equation \eqref{area} ensures that
\begin{equation}
 \int_{-1}^1 f(x,t)\dd x \le c.
\end{equation}
As $f$ is Lipshitz and vanishes on the boundary of $[-1,1]$ this implies that $f(x,t)\le \sqrt{c}$ uniformly and thus
that 
$$f^\delta(x,t)\le \sqrt{c}+\delta c/2.$$
This is enough to conclude, because $c$ can be chosen arbitrarily small.

\end{proof}

\subsection{Proof of Lemmata \ref{iltep} and \ref{iltep2}}

The final step in the proof of Theorem \ref{mainres} is the proof of Lemma \ref{iltep} and Lemma \ref{iltep2}.
For commodity reason, we chose to shift time from $\gep k$ to zero. 
What we need to show is that if one start with an initial condition $\eta^L_0$ that satisfies
\begin{equation}\label{grominet2}
\forall x\in[-1,1], \quad   \bar \eta^L_0(x)\ge f^{\delta}(x,k\gep),
\end{equation}
then w.h.p. 
\begin{equation}\label{lem1}
 \forall x\in[-1,1], \quad   \bar \eta(x,\gep)\ge f^{\delta}(x,(k+1)\gep),
\end{equation}
and 
\begin{equation}\label{lem2}
  \quad \forall t\in(0,\gep) \bar \eta(x,t)\ge f^{\delta}(x,k\gep)-\alpha.
\end{equation}
Note that the assumption that we suppose is not exactly the same as in the Lemmata \ref{iltep} and \ref{iltep2}
as \eqref{grominet} holds only w.h.p.\ but this is not a problem, as anything happening on a set of vanishing probability does not change the conclusion.

In the remainder of the proof, we shall consider only the case $k=0$, to lighten the notation.
The reader can check then that the bound we use throughout the proof are in fact valid uniformly for all $k$.
For instance for Proposition \ref{samere} we can note $\bar \delta(\gep k)$ is bounded from below by $\bar \delta c$, 
and for the rest \eqref{lesbornes} provides uniform bounds.

\medskip

By monotonicity, it is sufficient to prove \eqref{lem1} and \eqref{lem2}  starting from the smallest initial condition satisfying \eqref{grominet2}.
In in this case we habe
\begin{equation}\label{crounch2}
 \bar \eta^L_0(x)= f_0^\gd+\sigma_L(x),
\end{equation}
where $0\le \sigma_L(x)\le 2/L$.
\medskip

The strategy is then quite simple.
Set $\bar l=l_0-\bar\delta$ and $\bar r=r_0-\bar\delta$ to be the left and right boundaries of the region where
$f^{\delta}(\cdot,\gep k)$ is positive.
First, we show that for a small time $\gep$, the boundary of the unpinned region do not move more further than $\gep^2$ from their original location, i.e. that  w.h.p.\
$\eta$ does not add contact points with the wall in the interval $L[\bar l+\gep^2,\bar r-\gep^2]$.

\begin{proposition}\label{samere}
For $\gep$ small enough (depending only on $\delta$)
the dynamic started from an initial condition satisfying \eqref{crounch2}.
we have  w.h.p.\ 
\begin{equation}
 \bar \eta(x,t)>0, \quad \forall x\in
 [\bar l+\gep^2,\bar r-\gep^2 ], \forall t\in[0,\gep].
\end{equation}
\end{proposition}

The second step, at the end of the section, is then to say that if no contact is added, 
then the dynamics restricted to the interval $L[\bar l+\gep^2,  \bar r-\gep^2]$ stochastically
dominates  a corner-flip dynamics like the one of Section \ref{cornerflip} for which the scaling limit is given by the 
solution of the heat equation with an initial condition that is close to $f_0$.
What remains to do at last is to compare the solution of the heat-equation after time $\gep$ to $f(\cdot,\gep)$ to establish \eqref{lem1} and to 
$f_0-\alpha$ to establish \eqref{lem2}.

\begin{proof}[Proof of Proposition \ref{samere}]
The trick is to show that 
if one touches the wall to soon, the area below the curve $a(\bar \eta)$ shrinks too fast.
Set $\mathcal T'$ to  be the first time at which $\eta$ touches the wall in the interval $L[\bar l+\gep^2,\bar r-\gep^2]$ and 
$\tau'=\mathcal T'/L^2$.
We decompose the proof in two lemmata which proofs are postponed.

\medskip

The first step of our reasoning is to prove that Lemma \ref{drifta} is almost sharp up to time $\mathcal T'$. 
Let $ \tilde A(\eta)$ be defined as 
\begin{equation}
 \tilde A(\eta(t)):=\int_{x_l}^{x_r} \max(\eta(x,t),1)\dd x.
\end{equation}
where
$$x_l:=\lceil L (\bar l+\gep^2) \rceil \text{ and } x_r:=\lfloor L (\bar r-\gep^2)\rfloor.$$

\begin{figure}[hlt]
 \begin{center}
 \leavevmode 
 \epsfxsize =8 cm
 \psfragscanon
 \psfrag{1}{\small $\bar l$}
 \psfrag{-1}{\small $\bar r$}
\psfrag{d}{\small $\gep^2$}
 \epsfbox{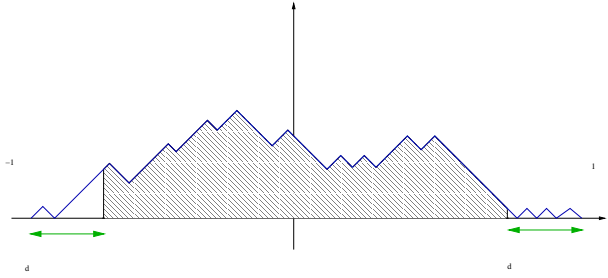}
 \end{center}
 \caption{\label{plusbas} A trajectory $\eta$, with the volume $ A^d(\eta)$ darkened.
The drift of the volume when there are no contact with the wall is either equal to $2$, $1$, $0$, $-1$ or $-2$.}
 \end{figure}

As in Lemma \ref{drifta}, one can compute explicitly $\mathcal L \tilde A$ and use the expression we obtain to prove,

\begin{lemma}\label{ksksks}
One has w.h.p.\
\begin{equation}\label{doobi}
 \tilde A(\eta(\cT'))\ge \tilde A(\eta_0)-2\cT'-L^{7/4},
\end{equation}
and as a consequence, w.h.p.
\begin{equation}\label{frete}
a(\bar \eta(\tau'))\ge a(f^{\delta}_0)-2\tau'-\gep^4-L^{-1/4}.
\end{equation}
\end{lemma}

The second step is to show that if $\tau'$ is to small, then we have a lower bound on $a(\bar \eta(\tau'))$ which brings a contradiction to Lemma \ref{ksksks}.
This bound is obtained by combining Proposition \ref{coucc} with a simple geometric argument.

\medskip

Set $X_{\tau'}$ to be the point where the first contact with the wall in $[x_l,x_r]$ is occurs. It is the only $x$ that satisfies
\begin{equation}
L X_{\tau'}\in [x_l,x_r] \text{ and } \bar \eta(X_{\tau'},\tau')=0.
\end{equation}
Set also $\tilde f^\delta$ to be the solution heat equation with Dirichlet boundary condition on $[\bar l,\bar r]$, and initial condition $f^\delta_0$.

\begin{lemma}\label{compra}
For the dynamics $\eta$ starting from an initial condition that satisfies \eqref{crounch2}.
For every positive $\alpha$, w.h.p.\
we have

\begin{equation}\label{aeez}
 \forall x\in [\bar l,\bar r], \quad \bar \eta(x,\tau')\le \min(\tilde f^\delta(x,\tau'),|x-X_{\tau'}|)+\alpha.
\end{equation}
In addition if $\tau'\le \gep$, and $\gep$ is small enough (depending on $\delta$ and $c$) one has 
\begin{equation}\label{michou}
\int_{-1}^1\min( \tilde f^\delta(x,\tau'),|x-X_{\tau'}|))\dd x \le a (f_0^{\delta})-2\tau'-\gep^2\bar \delta/16
\end{equation}
\end{lemma}

Suppose now that with a non vanishing probability $\tau'\le \gep$.
Then Lemma \ref{compra} implies (using \eqref{aeez} with $\alpha=\gep^2\bar \delta /64$ that  with  a non vanishing probability,
\begin{equation}
a(\bar \eta(\tau'))\le  a (f^{\delta}(\gep k))-2\tau'-\gep^2\bar \delta /32.
\end{equation}
If $\gep$ is chosen such that 
$\gep^2\le \delta c/64$ then this brings a contradiction to \eqref{frete} for $L$ sufficiently large.
\end{proof}

\begin{proof}[Proof of Lemma \ref{ksksks}]

Evaluating the effect of each transition on $\tilde A$, and noticing in particular that corner flips involving either $x_l$ or $x_r$ modifies 
$\tilde A$ only  by $\pm 1$, one obtains
\begin{multline}
\mathcal L \tilde A(\eta):=2 \big(|\{ x\in (-x_l,x_r) \ | \ \eta_x\ge 1, \eta_{x\pm 1}=\eta_x+1\}|\\-
|\{ x\in (-x_l,x_r) \ | \ \eta_x\ge 2, \eta_{x\pm 1}=\eta_x-1\}|
\big)\\
-\ind_{\{\eta_{x_l}\ge 2, \ \eta_{x_l}=\eta_{x_l\pm 1}-1\}}+\ind_{\{\eta_{x_l}\ge  1,\ \eta_{x_l}=\eta_{x_l\pm 1}+1\}}\\
-\ind_{\{\eta_{-x_r}\ge 2,\ \eta_{x_r}=\eta_{x_r\pm 1}-1\}}+\ind_{\{\eta_{x_r}\ge 1,\ \eta_{x_r}=\eta_{x_r\pm 1}+1\}}.
\end{multline}
It is then easy to check that $t<\mathcal T'$
\begin{equation}\label{gronaz}
\mathcal L \tilde A(\eta(t))\in\{2,1,0,-1,-2\}.
\end{equation}
It is just a consequence of the fact that the difference between the number of local maximum and local minimum in $[x_l,x_r]$ is at most $1$.

Next, as in Lemma \ref{drifta}, we use Doobs maximal inequality for the martingale
\begin{equation}
  \tilde M_t:=\tilde A(\eta(t))- \tilde A(\eta_0)-\int_0^t \mathcal L \tilde A(\eta(s))\dd s, 
\end{equation}
and obtain that
\begin{equation}\label{zaq}
 \bbE[\min_{t\in[0,L^2]} \tilde M_t\le -L^{7/4}]\le L^{-7/2}\bbE\left[\tilde M^2_{L^2}\right].
\end{equation}
As $\tilde M_0=0$, the expectation of $\tilde M^2_{L^2}$ is equal that of the martingale bracket
$\langle \tilde M^2\rangle_{L^2}$, which can be shown to be almost surely bounded by $8L^3$ (recall \ref{vardrift}). It implies that 
the r.h.s.\ of \eqref{zaq} vanishes when $L$ tends to infinity.
Next, we notice that

\begin{equation}
 \bbP[\tilde M_{\mathcal T'}<-L^{7/4}]
 \le \bbP[ \mathcal T'<L^2]+ \bbE[\min_{t\in[0,L^2]} \tilde M_t\le -L^{7/4}],
\end{equation}
and we note that by \eqref{croci} (and the fact that $\cT'\le \cT$), the first term in the right-hand side also vanishes.
Hence, we have w.h.p.\
\begin{equation}
 \tilde A(\eta(\cT'))\ge \tilde A(\eta_0)-\int_0^\cT \mathcal L \tilde A(\eta(s))\dd s-L^{7/4},
\end{equation}
and \eqref{gronaz} allows us to obtain \eqref{doobi}.

\medskip

For the second-point we first notice that from the various definitions we have
\begin{equation}\label{zeri}
 a(\bar \eta(\tau'))\ge \frac{1}{L^2} \tilde A( \eta(\mathcal T')),
\end{equation}

Then, noticing that  the area below the curve in $[-L,x_l]\cup[x_r,L]$ cannot be larger $\gep^4$ (the curve is $1$-Lipshitz), we have
\begin{equation}\label{teri}
\tilde A(\eta_0)\ge L^2 (a(f^{\delta}_0)-\gep^4),
\end{equation}
and hence \eqref{frete} holds by a combination of \eqref{doobi}, \eqref{zeri} and \eqref{teri}.
\end{proof}

\begin{proof}[Proof of Lemma \ref{compra}]

The fact that $\bar \eta(x,\tau')\le |x-X_\tau'|$ is just derived from the fact that $\bar \eta$ is a Lipshitz function which equals zero at $X_\tau'$.
The inequality 
$$\bar \eta(x,\tau')\le  \tilde f^\delta(x,\tau')+\alpha, \quad \forall x\in(-1,1)$$ is derived from Proposition \ref{coucc}, but one has to be careful since 
$\tau'$ is a random time.
One has
\begin{equation}
\bbP[\max_{x\in[\bar l, \bar r]} [\bar \eta(x,\tau')- \tilde f^\delta(x,\tau')]\ge \alpha]
\le \bbP[\tau'\ge 1]+\bbP\left[\maxtwo{x\in[\bar l, \bar r]}{t\in[0,1]} [\bar \eta(x,t)- \tilde f^\delta(x,t)]\ge \alpha\right].
\end{equation}
The first term has vanishing probability from \eqref{croci} and the second one from Proposition \ref{coucc}.

\medskip

The second point \eqref{michou} is a bit more technical.
The first task is to show is that if $\tau'\le \gep$

\begin{equation}\label{minots}
\int_{-1}^1\min(\tilde f^{\delta}(x,\tau'),|x-X_{\tau'}|))\dd x \le \int_{-1}^1\tilde f^{\delta}(x,\tau')-\frac{\gep^2\bar \delta}{8}.
\end{equation}

We consider separately two cases, either $X_{\tau'}$ is far from the boundary say $\min(|X_{\tau'}-\bar l |,X_{\tau'}-\bar r|)\ge \bar\delta/4$.
By symmetry we can suppose that $X_{\tau'}$ is closer to the left boundary.
A consequence of \eqref{wercs3} in Lemma \ref{pasdrol} is that 
\begin{equation}
x-X_{\tau'} \le  \tilde f^{\delta}(x,\tau')-\bar \delta/16, \quad \forall x\in (X_{\tau'},X_{\tau'}+\bar \delta/16),
\end{equation}
which implies that
\begin{equation} \label{wert}
\int_{-1}^1\min(\tilde f^{\delta}(x,\tau'),|x-X_{\tau'}|))\dd x \le \int_{-1}^1\tilde f^{\delta}(x,\tau')-\frac{\bar \delta^2}{256}.
\end{equation}
We can say that \eqref{minots}  holds for all value of $X_{\tau'}$ provided that $\gep^2\le \bar \delta/32$. 

\medskip

If $X_{\tau'}$ is close to one of the boundary say $X_{\tau'}\in [\bar l+\gep^2,\bar l+\bar \delta/4]$, then \eqref{wercs2} implies that for all $x\le \in[X_\tau, X_\tau+ \bar \delta/4]$
\begin{multline}
x-X_{\tau'}\le \tilde f^{\delta}(x,\tau')-(x-\bar l)(1-e^{-\frac{\bar \delta^2}{16t}})+(x-X_{\tau'}).
 \\ \le \tilde f^{\delta}(x,\tau')-(X_{\tau'}-\bar l)+(x-\bar l)e^{-\frac{\bar \delta^2}{16t}}
\\ \le  \tilde f^{\delta}(x,\tau')-\gep^2+(\bar \delta/2)e^{-\frac{\bar \delta^2}{16\gep}}\le \tilde f^{\delta}(x,\tau')-\gep^2/2.
\end{multline}
The last inequality is valid if $\gep$ is small enough (how small depending on  $\bar \delta$). 
Integrating the above inequality over $[X_\tau, X_\tau+ \bar \delta/4]$ we obtain \eqref{minots} also in that case

\medskip

To conclude, it is then sufficient to use \eqref{wercs} in the r.h.s\ of \eqref{minots} with $\gep$ sufficiently small.

 \begin{figure}[hlt]
 \begin{center}
 \leavevmode 
 \epsfxsize =8 cm
 \psfragscanon
 \psfrag{d}{\small $\gep^2$}
 \psfrag{x_d}{\small $X_\tau'$}
 \epsfbox{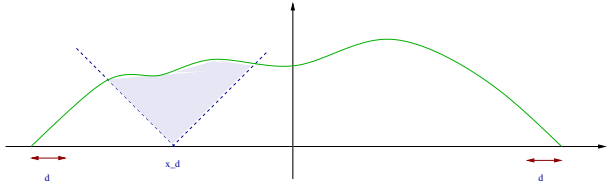}
 \end{center}
 \caption{\label{diffdaire} Figure representing the two function $\tilde f^{\delta}(\cdot,\tau')$, and $|\cdot-X_\tau'|$. The difference of volume below the graph of
$\tilde f^{\delta}(\cdot,\tau')$ and $\min(\tilde f^{\delta}(\cdot,\tau'),|\cdot-X_\tau'|)$ is the dark area on the figure. We prove a lower-bound on it using that
$\tilde f^{\delta}$ is not too small if $X_{\tau'}$ lies in the middle of the interval, or that the slope of $f$ on the boundary is close to one if $X_{\tau'}$ is closer
to the boundary.}
 \end{figure}

\end{proof}
%
%
%
%
%
%
%
%
%

We just proved that the dynamic do not touch the wall in the interval $[\bar l+\gep^2,\bar r-\gep^2]$. This allows us 
to compare it with the dynamic with no-wall for which we know the scaling limit by Theorem \eqref{frefer}.
If one runs the dynamic up to a time $\gep$, according to Proposition \ref{samere} 
the dynamic $\eta$ coincide w.h.p.\ with a modified one $\eta^{(\gep)}$ where there is no 
wall-constraint in the interval $[x_l,x_r]$.

\medskip

Using monotonicity of the dynamics, this second dynamic can be bounded from below 
by a dynamic with the domain reduced to $[x_l,x_r]$ and with a modified initial condition
which satisfies
$$\forall x\in [x_l,x_r], \quad \eta^{(\gep)}_0(x)\le \eta_0(x).$$
As $\eta_0$ satisfies \eqref{crounch2}
we can choose a sequence of initial condition $\eta^{(\gep)}_0$ which satisfies

\begin{equation}
\bar \eta^{(\gep)}_0(x):= \frac{1}{L}\eta^{(\gep)}_0(Lx)=f^{(\gep,\delta)}_0+o(1),
\end{equation}
where $f^{(\gep,\delta)}_0$ is defined on $[\bar l+\gep^2,\bar r-\gep^2]$ by

\begin{equation}
f^{(\gep,\delta)}_0:=f_0^{\delta}(x)-\gep^2.
\end{equation}
One calls $\bar \eta^{(\gep)}$ the resulting space-time rescaled dynamics.
By Theorem \ref{frefer} w.h.p.\

\begin{equation}\label{spce}
 \forall t>0 \forall x\in[-1,1],\quad \bar \eta^{(\gep)}(x,t)=\tilde f^{(\gep,\delta)}(x,t)+o(1),
\end{equation}
where $\tilde f^{(\gep,\delta)}(x,t)$ denotes the solution of the heat equation on $[\bar l+\gep^2,\bar r-\gep^2]$ with Dirichlet boundary condition and 
initial condition $f^{(\gep,\delta)}_0$.
Thus, to prove that Equation \eqref{lem1} and \eqref{lem2} hold for $k=0$, it is sufficient to show that

\begin{equation}\label{zabe}
  \forall x\in[l(\gep)-\bar \delta(\gep),r(\gep)+\bar \delta(\gep)],\quad \tilde f^{(\gep,\delta)}(x,\gep)>f^{\delta}(x,\gep),
\end{equation}
and that
\begin{equation}\label{crouvile}
  \forall x\in [\bar l+\gep^2,\bar r-\gep^2], \forall t\in [0,\gep], \quad  \tilde f^{(\gep,\delta)}(x,t)>f^{\delta}_0(x)-\alpha.
\end{equation}

\medskip

The second inequality is easier, it is a consequence of the fact that $\partial_t \tilde f^{(\gep,\delta)}= \tilde f^{(\gep,\delta)}_{xx}$ is bounded from below by 
$-\|f''_0\|_\infty$ because the minimum of the solution of the heat-equation with Dirichlet boundary condition is increasing.
Thus \eqref{crouvile} holds provided 
$$\gep\|f''_0\|_\infty<\alpha.$$
To prove  \ref{lem2} for $k>0$, we use \eqref{lesbornes} and replace $\|f''_0\|_\infty$ by $C$.

\medskip

Now we turn to the proof of \eqref{zabe}.
To control the value of $f^{\delta}(x,\gep)$, we use \eqref{clep}, \eqref{lesbornes} and the Taylor-Young formula to obtain

\begin{equation}\label{teces1}
 f^{\delta}(x,\gep)\le \begin{cases}
                     f^{\delta}_0(x)+\gep(f''_0(x)-\delta)+C\gep^2, \quad \forall x\in (l_0,r_0),\\
f^{\delta}_0(x)+\gep( f_0 ''(l(0))-\delta)+C \gep^2, 
 \quad  \forall x\in [l(\gep)-\bar\delta(\gep),l_0]\\
 f^{\delta}_0(x)+\gep(f_0''(r(0))-\delta)+C \gep^2, \quad \forall x\in [r_0,r(\gep)+\bar\delta(\gep)],
                    \end{cases}
\end{equation}
and
\begin{equation}\label{teces2}
\begin{split}
 l(\gep)-l_0&
 \ge\gep (-f_0''(l_0)+\bar \delta)-C\gep^2,\\
  r(\gep)-r_0&\le \gep(f_0''(r_0)-\bar \delta))+C\gep^2.
\end{split}\end{equation}

Controlling $ \tilde f^{(\gep,\delta)}(x,\gep)$ is more tedious. As the initial condition is not $C^2$, there is no continuity of the 
time derivatives. Let $K$ denote the heat kernel on $I:=[\bar l+\gep^2, \bar r-\gep^2]$ with Dirichlet boundary condition.
We have

\begin{equation}\label{brin}
 \tilde f^{(\gep,\delta)}(x,\gep)=f^{(\delta)}_0(x)-\gep^2+\int^\gep_0  \tilde f^{(\gep,\delta)}_{xx}(x,t) \dd t
\end{equation}
and 
\begin{equation}
\tilde f^{(\gep,\delta)}_{xx}(x,t)=\int_I f''_0(y)\ind_{\{y\in [l_0,r_0]\}}K(x,y,t)\dd y.
\end{equation}
Now for $t\le \gep$, and if $\gep$ is sufficiently small (how small depending on the value of $\bar \delta$) we have uniformly in $x\in (l_0,r_0)$ 
\begin{equation}
\int_{|y-x|\le \gep^{1/3}} K(x,y,t) \dd x \ge 1-\gep,
\end{equation}
and from \eqref{lesbornes} for  $x\in (l_0,r_0)$ 
\begin{equation}
 f''_0(y)\ind_{\{y\in [l_0,r_0]\}}\le f''_0(x)+C|x-y|.
\end{equation}
This implies that for  $x\in (l_0,r_0)$ 
\begin{equation}\label{groums}
 \tilde f^{(\gep,\delta)}_{xx}(x,t)\ge f''_0(x)-C(\gep+\gep^{1/3}).
\end{equation}
For $x\notin (l_0,r_0)$,  using similar heat kernel estimates, we obtain that for all $t\in[0,\gep]$
\begin{equation}\label{groums2}
  \tilde f^{(\gep,\delta)}_{xx}(x,t)\ge\begin{cases} f''_0(l_0)-C\gep^{1/3} \text{ for } x\in (l_0-\bar \delta+\gep^2, l_0),\\
  f''_0(r_0)-C\gep^{1/3} \text{ for } x\in (r_0, r_0+\bar \delta-\gep^2).
                                     \end{cases})]
\end{equation}
Note that neither \eqref{groums} nor \eqref{groums2} are optimal estimates but they are sufficient for our purpose.
Using these estimates in \eqref{brin} we obtain that

\begin{equation}
  \tilde f^{(\gep,\delta)}(x,\gep)\ge \begin{cases}
                                       f^{(\delta)}_0(x)-\gep^2-\gep f''_0(x)-2C\gep^{4/3} \text{ for } x\in (l_0,r_0),\\
                                       f^{(\delta)}_0(x)-\gep^2-\gep f''_0(l_0)-C\gep^{4/3}  \text{ for } x\in (l_0-\bar \delta+\gep^2,l_0),\\
                                         f^{(\delta)}_0(x)-\gep^2-\gep f''_0(r_0)-C\gep^{4/3}  \text{ for } x\in (r_0, r_0+ \bar \delta-\gep^2).
                                      \end{cases}
\end{equation}

Comparing the above inequalities with \eqref{teces1}, we can conclude that \eqref{zabe} holds. This ends the proof. 

\qed
%

{\bf Acknowledgement:} The author would like to thank Inwon Kim for her usefull advice for the proof of the short time existence of a solution to the Stefan Problem,
Jean Dolbeault for pointing out a simple proof for Lemma \ref{agmon}, Jimmy Lamboley for various discussions and Julien Sohier for his comments on the manuscript.
 This research work has been initiated during the authors stay in Instituto Nacional de Matématica Pura e Aplicada, he acknowledges 
 gratefully hospitality and support.

\appendix

\section{}

In this section, we prove several technical results.
The first one concerns the time needed for a corner-flip dynamics to reach an atypically low position.

\begin{lemma}\label{mousis}
Consider a corner-flip dynamics on $\gO^0_L$, started from an initial condition that satisfies

\begin{equation}
\eta_0(x)\ge \min(x+L,-x+L, L^{3/4}).
\end{equation}
Then with large probability, for all $t\le\exp(L^{1/4})$
\begin{equation}
\tilde\eta(x,t)\ge -L^{3/4}, \forall x\in[-L,L]
\end{equation}

\end{lemma}

\begin{proof}
One couples $\tilde \eta$ with a dynamic $\tilde \eta^{2}$ starting from the uniform measure on $\gO^{0}_{L}$ (we denote it by $\pi$).
Because of our choice for the initial condition of $\tilde\eta$ is above $\tilde\eta^{2}$ at time zero with large probability 
and thus we can couple the two dynamics so that $\tilde\eta\ge\tilde\eta^{2}$ for all time with large probability.

\medskip

Hence it is sufficient to prove the result for $\tilde\eta^{2}$.
Consider the discrete-time dynamics $\hat\eta^2(n)$ starting from the uniform measure and that at each step choses $x$ at random and flip $\eta$ to $\eta^{(x)}$.
As $\pi$ is left invariant by this dynamics, the probability that $\hat\eta(x,n)< -L^{3/4}$ for some $x$ after $\exp(2L^{1/4})$ 
step is at most (by union bound)
\begin{equation}
 \exp(2L^{1/4})\pi(\exists x\in[-L,L], \eta(x)< -L^{3/4})=O(\exp(-L^{1/2})).
\end{equation}
The dynamics in continuous time can then be construction from $\hat \eta$ by considering $(\tau_n)_{n\ge0}$ a Poisson clock process of intensity $L$
and setting $\tilde \eta^2(t)=\hat \eta(n)$ if $t\in[\tau_n,\tau_{n+1})$. Then we conclude by remarking that 
with high probability 
\begin{equation}
 \tau_{\exp(2L^{1/4})}\ge \exp(L^{1/4}).
\end{equation}

\end{proof}

The second result concerns estimate on $\tilde f^{\delta}$, the solution of the heat-equation with initial condition $f^{\delta}_0$ and Dirichlet boundary condition
on $[\bar l,\bar r]$. It says that in the time interval $[0,\gep]$ the slope of $\tilde f^{\delta}$ near the boundary stays close to one.

\begin{lemma}\label{pasdrol}
For $\gep$ small enough, for all $t\le \gep$ the following three statements hold
\begin{itemize}
\item[(i)] For all $s\in [0,\bar \delta/2]$,
\begin{equation}\label{wercs2}
\max\left(  \tilde f^{\delta}(\bar l+s,t) ,\tilde f^{\delta}(\bar r-s,t)\right)\ge s(1-e^{-\frac{\bar \delta^2}{16t}}).
\end{equation}
\item[(ii)] 
For all $x\in [\bar l+\bar \delta/4, \bar r -\bar \delta/4]$
\begin{equation}\label{wercs3}
 \tilde f^{\delta}(x,t)\ge \bar \delta/8.
 \end{equation}
\item[(iii)] We also have
\begin{equation}\label{wercs}
\int_{\bar l}^{\bar r} \tilde f^{\delta}(x,t)\dd x \ge
\int_{\bar l}^{\bar r} f^\delta_0(x)\dd x-2t+e^{-\frac{\bar \delta^2}{16t}}.
\end{equation}
\end{itemize}

\end{lemma}

\begin{proof}
The important point is to control the value of $f_x$ near $\bar l$ and $\bar r$ and the rest follows.
We have
\begin{equation}
\tilde f^{\delta}_x(x,t)=\bE_x[(f^{\delta}_0)'(B_t)],
\end{equation}
where $\bE_x$ is the expectation with respect to standard Brownian Motion reflected on the boundaries of $[\bar l, \bar r]$.
As $f^{\delta}_0)' \in[-1,1]$ in the whole interval  and is equal to $1$ on $[\bar l, \bar l+\bar \delta]$ one has

\begin{equation}
\bE_x[(f^{\delta}_0)'(B_t)]\ge 1-\bP_x[B_t \ge\bar l+\delta].
\end{equation}
Finally if $t$ is much smaller than $\bar \delta^2$ and $x\in [\bar l, \bar l+\bar \delta/2]$,
\begin{equation}
\bP_x[B_t \ge \bar l+\bar \delta]\le e^{-\frac{\bar \delta^2}{16t}},
\end{equation}
which implies that for all $x\in [\bar l, \bar l+\bar \delta/2]$,
 \begin{equation}
 \tilde f^{\delta}_x(x,t)\ge 1-e^{-\frac{\bar \delta^2}{16t}}.
\end{equation}
Similarly for all $x\in [\bar r-\bar \delta/2,\bar r]$
\begin{equation}
 \tilde f^{\delta}_x(x,t)\le -1+e^{-\frac{\bar \delta^2}{16t}}.
\end{equation}
Integrating these inequalities  we obtain \eqref{wercs2}.

\medskip

For $(ii)$, we notice that \eqref{wercs3} is a consequence of the first point for $x\in \{\bar l+\bar \delta/4, \bar r -\bar \delta/4\}$.
For points inside the interval, it is sufficient to notice that local minima of the solution of the heat equation increase, and that they are initially larger than $\bar \delta$.

\medskip

For $(iii)$ we notice that 
\begin{equation}
\partial_s\left(\int_{\bar l}^{\bar r} \tilde f^{\delta}(x,s)\dd x\right)= \tilde f^{\delta}_x(\bar l,t)-  \tilde f^{\delta}_x(\bar r,t)
\ge -2+2e^{-\frac{\bar \delta^2}{16s}},
\end{equation}
if $s$ is much smaller than $\bar \delta^2$. Then \eqref{wercs} follows by integrating over $s\in[0,t]$.
\end{proof}

\end{document}